\RequirePackage{amsthm}

\documentclass[sn-mathphys-num]{sn-jnl}

\usepackage{anyfontsize}
\usepackage{graphicx}%
\usepackage{multirow}%
\usepackage{amsmath,amssymb,amsfonts}%
\usepackage{amsthm}%
\usepackage{mathrsfs}%
\usepackage[title]{appendix}%
\usepackage{xcolor}%
\usepackage{textcomp}%
\usepackage{manyfoot}%
\usepackage{booktabs}%
\usepackage{algorithm}%
\usepackage{algorithmicx}%
\usepackage{algpseudocode}%
\usepackage{listings}%
\usepackage{lmodern}
\usepackage{amsmath}
\usepackage{amsthm}
\usepackage{amssymb}
\usepackage{amsfonts}
\usepackage{graphicx}
\usepackage[inline]{enumitem}
\usepackage{multirow}
\usepackage{mathtools}
\usepackage{etextools}
\usepackage{xcolor}
\usepackage{subcaption}

\def \erre{\mathbb{R}}

\def \PMp{\mathcal{PM}_{\vec v}}
\def \PMp{\mathcal{PM}_{\vec v}}

\def \EE{\mathbb{E}}

\def \erre{\mathbb{R}}

\DeclarePairedDelimiter\ceil{\lceil}{\rceil}
\DeclarePairedDelimiter\floor{\lfloor}{\rfloor}

\newtheorem{mechanism}{Mechanism}

\theoremstyle{thmstyleone}%
\newtheorem{theorem}{Theorem}
\usepackage[symbol]{footmisc}
\theoremstyle{thmstyletwo}%
\newtheorem{example}{Example}%
\newtheorem{remark}{Remark}%

\theoremstyle{thmstylethree}%
\newtheorem{definition}{Definition}%

\begin{document}

\title[Article Title]{Mechanism Design for Locating Facilities with Capacities with Insufficient Resources\footnote[2]{A previous version of this paper appeared in the proceedings of UAI24, \cite{auricchio2024on}. In this improved version, we study the problem in higher dimension and enhance our numerical results. }}

\author*[1]{\fnm{Gennaro} \sur{Auricchio}}\email{gennaro.auricchio@unipd.it}

\author[2]{\fnm{Harry J.} \sur{Clough}}\email{hc2405@bath.ac.uk}

\author[2]{\fnm{Jie} \sur{Zhang}}\email{jz2558@bath.ac.uk}

\affil[1]{\orgdiv{Department of Mathematics}, \orgname{Università di Padova}, \city{Padua}, \state{Italy}}

\affil[2]{\orgdiv{Department of Computer Science}, \orgname{University of Bath},  \city{Bath}, \country{UK}}


\abstract{This paper explores the Mechanism Design aspects of the $m$-Capacitated Facility Location Problem where the total facility capacity is less than the number of agents. 
Following the framework outlined in \cite{aziz2020capacity}, the Social Welfare of the facility location is determined through a First-Come-First-Served (FCFS) game, in which agents compete once the facility positions are established. 
When the number of facilities is $m > 1$, the Nash Equilibrium (NE) of the FCFS game is not unique, making the utility of the agents and the concept of truthfulness unclear. 
To tackle these issues, we consider absolutely truthful mechanisms—mechanisms that prevent agents from misreporting regardless of the strategies used during the FCFS game. 
We combine this stricter truthfulness requirement with the notion of Equilibrium Stable (ES) mechanisms, which are mechanisms whose Social Welfare does not depend on the NE of the FCFS game. 
We demonstrate that the class of percentile mechanisms is absolutely truthful and identify the conditions under which they are ES. 
We also show that the approximation ratio of each ES percentile mechanism is bounded and determine its value. 
Notably, when all the facilities have the same capacity and the number of agents is sufficiently large, it is possible to achieve an approximation ratio smaller than $1+\frac{1}{2m-1}$.
Finally, we extend our study to encompass higher-dimensional problems.
Within this framework, we demonstrate that the class of ES percentile mechanisms is even more restricted and characterize the mechanisms that are both ES and absolutely truthful.
We further support our findings by empirically evaluating the performance of the mechanisms when the agents are the samples of a distribution.}

\keywords{Facility Location, Mechanism Design, Game Theory, Nash Equilibrium}



\maketitle

\section{Introduction}

The $m$-Capacitated Facility Location Problem ($m$-CFLP) extends the $m$-Facility Location Problem ($m$-FLP) by introducing a constraint that bounds the number of clients the facility can accommodate \cite{brimberg2001capacitated,pal2001facility,aardal2015approximation}.
Both the $m$-FLP and the $m$-CFLP are essential subproblems in various applications within social choice theory.
These include areas such as disaster relief \cite{doi:10.1080/13675560701561789}, where efficient allocation of resources is of critical importance; supply chain management \cite{MELO2009401}, which focuses on optimizing logistics and distribution networks; healthcare systems \cite{ahmadi2017survey}, aiming to improve service delivery and accessibility; clustering techniques \cite{hastie2009elements,auricchio2019computing}, used in data analysis and machine learning; and public facility accessibility \cite{barda1990multicriteria}, which enhances community services by effectively placing facilities to serve the population.
In its fundamental guise, the $m$-CFLP consists of determining the location of $m$ facilities starting from the positions of $n$ agents.
Each facility has a capacity limit, which describes the maximum amount of agents it can serve.
While the algorithmic aspects of this problem have been extensively studied in the literature (see \cite{brimberg2001capacitated}), the mechanism design aspects of the problem have only recently started to garner the attention of the scientific community.
In mechanism design, the Facility Location Problem and the $m$-CFLP are studied under the assumption that each agent incurs a cost to access the facilities. 
This cost is usually equal to the distance between the agent and the nearest facility.
Consequently, every agent prefers to have a facility located as close to them as possible.
In this case, if each agent is in charge of reporting its own position on the line, optimizing a communal goal is subject to manipulation driven by the agents' self-interested behaviour. 
For this reason, a key property that a mechanism should possess is \textit{truthfulness}, which ensures that no agent can gain an advantage by misrepresenting their private information. 
This stringent property, however, forces the mechanism to produce suboptimal locations, leading to an efficiency loss which is quantified by the \textit{approximation ratio} -- that is the worst-case ratio between the objective achieved by the mechanism and the optimal objective attainable \cite{nisan1999algorithmic}.
Defining efficient routines that forbid agents from manipulating is the defining issue of mechanism design.
In this paper, we study the mechanism design aspects of the $m$-CFLP.
In particular, we focus on the framework presented in \cite{aziz2020capacity}, and extend it to the case in which we have $m$ facilities whose total capacity is lower than the number of agents needing accommodation.
In this framework, the mechanism designer cannot force agents to use a specific facility, therefore the agents compete in a First-Come-First-Served (FCFS) game to determine who is accommodated by the facilities.
The overall process therefore consists of two parts.
First, the agents report their position to a mechanism, which locates the facilities.
Second, the agents compete in the FCFS game to determine their utilities.
Notice that the Social Welfare achieved by the mechanism and the utilities of the agents depend on the Nash Equilibria (NE) of the FCFS game induced by the facilities' placements.
When we need to place a single facility, that is $m=1$, the FCFS game has always a unique NE.
When $m\ge 2$, however, the NE of the FCFS game is no longer unique, posing a series of challenges: if the NE is not unique, the agents' utilities and thus the Social Welfare achieved by a facility placement have different values depending on the equilibrium of the FCFS game.
As a consequence, the approximation ratio also depends on the specific NE.
Furthermore, the classic definition of the truthful mechanism is no longer suitable for this problem as it does not consider the different strategies that the agents may adopt during the FCFS game.
Addressing these issues are a major challenge for this problem, and thus they were left as an open problem in \cite{aziz2020capacity}.
A previous version of this paper appeared in the proceedings of UAI 2024.
In this improved version, we fully characterize the set of feasible percentile mechanisms by studying the All-In-One and Side-By-Side mechanism and extend our study to higher dimensional problems.
Furthermore, we enhance our numerical results by adding new metrics to study how well the best percentile mechanisms perform when the agents are assumed to be sampled from a probability distribution.

\subsection*{Our Contribution}

In this paper, we study the mechanism design aspects of the $m$-CFLP when the total capacity of the facilities is less than the number of agents.
In particular, we extend the framework presented in \cite{aziz2020capacity} to encompass problems in which there is more than one capacitated facility to locate.
First, we show that, regardless of how we locate the facilities, the FCFS game induced by the location has at least one pure NE.
We then present a notion of truthfulness that accounts for the different strategies the agents can adopt during the FCFS game, which we name \textit{absolute truthfulness}.
Finally, we introduce the class of Equilibrium Stable (ES) mechanisms, i.e. mechanisms whose output induces a FCFS game in which every NE achieves the same Social Welfare. 
Within this framework, we study the percentile mechanisms \cite{sui2013analysis}.
We show that every percentile mechanism is  absolutely truthful and then fully characterize the set of conditions under which a percentile mechanism is ES and compute its approximation ratio.
First, we consider the case $m=2$ and show that an absolutely truthful and ES percentile mechanism exists.
In particular, we show that there are three categories of ES percentile mechanisms: 
\begin{enumerate*}[label=(\roman*)]
\item the All-In-One (AIO) mechanisms, which place all the facilities at the same position, i.e. $y_1=y_2=x_i$; 
\item the Side-By-Side (SBS) mechanisms, which places all the facilities at two consecutive agents' positions, i.e. $y_1=x_{i}$ and $y_{2}=x_{i+1}$; and 
\item the Wide-Gap (WG) mechanisms, which places the facilities at two well-separated positions, i.e. $y_1=x_i$ and $y_2=x_j$ where $i+1<j$.
\end{enumerate*}
We then characterize the approximation ratio of a percentile mechanism as a function of the facilities' capacities and the vector inducing the percentile mechanism.
Consequentially, we determine the best percentile mechanism tailored to the number of agents $n$ and the capacities of the facilities $k_1$ and $k_2$.
We show that the best approximation ratio that an ES percentile mechanism placing two facilities can achieve is $\frac{4}{3}$, which occurs when $k_1=k_2=k$ and $n\ge 3k$.
We then consider the case in which $m>2$.
In this framework, a percentile mechanism that is ES and places the facilities at more than two different locations might not exist.
However, when all the facilities have the same capacity and $n\ge (2k-1)m$ holds, there exists a  percentile mechanism whose approximation ratio is less than $1+\frac{1}{m-1}$.
This result has two interesting consequences: \begin{enumerate*}[label=(\roman*)]
    \item it shows that, under suitable assumptions, the percentile mechanisms are asymptotically optimal with respect to $m$ and
    \item it highlights the differences between the $m$-CFLP and the $m$-FLP. Indeed, in the classic framework, any percentile mechanism has unbounded approximation ratio whenever $m>2$, \cite{walsh2020strategy,fotakis2014power}.
\end{enumerate*} 
We then extend our study to higher-dimensional problems. 
We demonstrate that when agents are supported in a space that has two or more dimensions, the set ES percentile mechanisms becomes even more restricted.
Specifically, a percentile mechanism is ES if and only if it places all the facilities at the same location or distributes all the facilities between two different points. 
Despite their differences, the worst-case analysis of these two types of percentile mechanisms yield similar worst-case guarantees.
Lastly, we empirically study the behaviour of the best percentile mechanisms under the assumption that the agents are distributed according to a distribution $\mu$.
We focus on the case in which we have two facilities, since it is the case in which the gap between $1$ and the approximation ratio of the best possible percentile mechanism is largest.
We conduct our analysis using two different metrics: the \textit{Bayesian approximation ratio}, which compares the expected cost of the mechanism and the expected optimal cost \cite{hartline2009simple} and the \textit{average-case ratio}, which measures the average ratio between the mechanism cost and the optimal cost \cite{DBLP:journals/iandc/DengGZ22}.
From our analysis, we observe that, regardless of the metric, when the agents follow a distribution, the performances of the ES mechanism are close to optimal, regardless of the distribution.

\subsection*{Related Work}

The Mechanism Design aspects of the $m$-FLP were firstly studied in \cite{procaccia2013approximate}, where the authors studied the problem of placing a set of facilities amongst $n$ self-interested agents who want to a facility located as close as possible to their real position.
In this pioneering work, the authors studied the problem of placing a facility that minimizes the sum of all the agents costs.
Following this seminal work, various mechanisms with constant approximation ratios for placing one or two facilities on lines \cite{DBLP:journals/aamas/Filos-RatsikasL17},  trees \cite{DBLP:conf/sigecom/FeldmanW13,DBLP:conf/atal/FilimonovM21}, circles \cite{DBLP:conf/sigecom/LuSWZ10,DBLP:conf/wine/LuWZ09}, general graphs \cite{10.2307/40800845,DBLP:conf/sigecom/DokowFMN12}, and metric spaces \cite{DBLP:conf/sagt/Meir19,DBLP:conf/sigecom/TangYZ20} were introduced.
Crucially, all these positive results are limited to scenarios where the number of agents is restricted or the number of facilities is either $1$ or $2$ and the objective is to minimize the total agents costs.
For a comprehensive survey of the mechanism design aspects of the FLP, we refer to \cite{chan2021mechanism}.
The $m$-Capacitated Facility Location Problem ($m$-CFLP) is a variation of the $m$-FLP in which each facility has a capacity limit.
The first game theoretical framework for the $m$-CFLP was presented in \cite{aziz2020facility}.
In this paper, the authors studied the case in which there are at least two facilities whose total capacity is enough to accommodate all the agents and the mechanism designer has to decide where to place the facilities and which agent can use them, so that the mechanism must elicit the positions and the agent-to-facility matching.
Following this initial study, in \cite{ijcai2022p75} the authors proposed a more theoretical analysis of the problem, while in \cite{auricchio2024facility} it was shown that it is possible to define deterministic mechanisms with bounded approximation ratio when all the facilities have the same capacities and the number of agents is equal to the total capacities of the facilities.
Lastly, the $m$-CFLP has been studied from a Bayesian mechanism design perspective in \cite{auricchio2023extended}.
In this paper, we consider an alternative game theoretical framework for the $m$-CFLP, firstly introduced in \cite{aziz2020capacity}.
This framework differs from the one proposed in \cite{aziz2020facility} for two main reasons:
\begin{enumerate*}[label=(\roman*)]
    \item the total capacity of the facilities is lower than the total number of agents, thus part of the agents cannot be accommodated and
    \item the mechanism designer does not enforce an agent-to-facility assignment. Thus, after the positions of the facilities are elicited, the agents compete in a First-Come-First-Served (FCFS) game to access the facilities.
\end{enumerate*}
When $m=1$, the FCFS game is trivial as the agents accommodated by the facility are the ones that are closer to the facility.
When $m>1$, designing mechanisms becomes much more complicated as, for example, the Nash Equilibrium (NE) of the FCFS game is no longer unique.
As a consequence the utility of every agent depends on the specific NE of the game, making the classic notion of trustfulness unfit for this framework.

\section{Setting Statement}
\label{sec:settingstatement}

Let $\vec x\in [0,1]^n$ be the position of $n$ agents in the interval $[0,1]$.
We denote with $\vec k=(k_1,\dots,k_m)$ the $m$-dimensional vector containing all the capacities of the $m$ facilities, so that $k_j$ is the maximum number of agents that the $j$-th facility can accommodate.
We assume that the total capacity of the facilities is less than the number of agents, hence $\sum_{j\in[m]}k_j<n$.
In this case, a mechanism is a function $M:[0,1]^n\to \erre^m$ that maps a vector containing the agents' reports to a facility location $\vec y=(y_1,\dots,y_m)$, where $y_j$ is the position of the facility with capacity $k_j$.
After the mechanism places the facilities, agents compete in a First-Come-First-Served (FCFS) game to access the facilities.

\subsection{First-Come-First-Served Game}
Let $\vec y=(y_1,\dots,y_m)$ be a vector containing the position of the facilities to locate, that is the facility with capacity $k_j$ is located at $y_j$.
In what follows, we implicitly assume that there is an internal ordering of the agents to break ties.
Then, given $\vec x\in [0,1]^n$ the vector containing the positions of the $n$ agents on the interval $[0,1]$, the FCFS game induced by the facility location $\vec y$ is as follows:
\begin{enumerate}[label=(\roman*)]
    \item Each agent selects one of the facilities, so that the set of strategies of each agent is the set $[m]:=\{1,2,\dots,m\}$.
    We denote with $\vec s\in[m]^n$ the vector containing a set of pure strategies. For every $\vec s$, we denote with $\mathcal{S}_j\subset [n]$ the set of agents that selected strategy $j$. 
    \item  Denoted with $d_{i,j}=|x_i-y_j|$ the distance of the agent $i$ from the location of the facility they selected, we define $T_j\subset \mathcal{S}_j$ as the set containing the agents in $\mathcal{S}_j$ whose value $d_{i,j}$ is in among the $k_j$ lowest values.
    Break ties according to the prefixed priority rule.
    \item Finally, the utility of agent $i$ is defined as follows
    \[
    u_i(\vec x,\vec y;\vec s)=
    \begin{cases}
        1-|x_i-y_j| \quad\quad \text{if} \quad  i\in T_j \\
        0 \quad\quad\quad\quad\quad\quad\quad \text{otherwise}
    \end{cases}.
    \]
\end{enumerate}

First, we show that every FCFS game has at least one pure Nash Equilibrium (NE).

\begin{theorem}
\label{thm:NEexistance}
    For every $\vec x\in[0,1]^n$, every $\vec y\in[0,1]^m$, and every capacity vector $\vec k$, the FCFS game associated with $\vec x$, $\vec y$, and $\vec k$ admits at least one pure Nash Equilibrium.    
\end{theorem}

\begin{proof}
Let $\vec x$ be the vector containing the position of the agents and let $\vec y$ be the position of the facilities.
We denote with $k_j$ the capacity of the facility located at $y_j$ for every $j\in[m]$.
In what follows, we assume that the set of agents has an inner ordering that decides how to break ties.
Let us define $\mathcal{D}$ the set containing all the distances agents to facility, that is $\mathcal{D}=\{|x_i-y_j|\}_{i\in[n],j\in[m]}$.
Let $\vec c,\vec s\in\erre^{m}$ be two null vectors, that is $c_j=s_j=0$ for every $j\in[m]$.
We now construct a Nash Equilibrium through the following iterative routine.
\begin{enumerate}
    \item Let $d$ be the minimum of the elements in $\mathcal{D}$.
    Up to a tie, there exist a couple $(i_1,j_1)\in[n]\times[m]$ such that $d=|x_{i_1}-y_{j_1}|$. 
    We set $c_{j_1}=c_{j_1}+1$, $s_{i_1}=j_{1}$, and remove all the elements of the form $|x_{i_1}-y_j|$ from $\mathcal{D}$.
    Then, if $c_{j_1}=k_{j_1}$, we remove from $\mathcal{D}$ all the elements of the form $|x_i-y_{j_1}|$.

    \item We repeat the routine of point $(1)$ until $\mathcal{D}$ becomes empty.

    \item If $s_i=0$ for some $i\in[n]$, we set them to be equal to $1$.
\end{enumerate}

Since $\mathcal{D}$ is discrete, the routine terminates in finite number of iterations and the output is a vector containing a set of agents' pure strategies.
We now show that the output of the routine $\vec s$ is a Nash Equilibrium by proving that no agent $i$ can increase its utility by deviating from playing $s_i$.
Toward a contradiction, assume that an agent $i$ can increase its utility by playing $s_i'$ rather than $s_i$.
By definition of $s_i$, we have that if $|x_i-y_{s_i'}|<|x_i-y_{s_i}|$, then there are at least $k_{s_i'}$ agents that are closer to $y_{s_i'}$ or that have a higher priority order than agent $i$ and play strategy $s_i'$.
Thus the agent cannot gain a benefit from deviating from $s_i$, which proves that $\vec s$ is a pure Nash Equilibrium.
\end{proof}

When the vector $\vec k$ is clear from the context, we denote the set of all the pure Nash Equilibria with $NE(\vec x,\vec y)$.
The Social Welfare (SW) of the facility location $\vec y$ according to $\gamma\in NE(\vec x,\vec y)$ is defined as the sum of all the agents' utilities, that is $SW_{\gamma}(\vec x,\vec y)=\sum_{i\in [n]}u_i(\vec x,\vec y;\gamma)$.
Notice that when $m=1$, the Nash Equilibrium of the FCFS game is unique, since every agent can play only one strategy, hence the SW of the game is well defined.
This is no longer true when we need to place more than one facility.
In this case, the SW of the game changes depending on the specific Nash Equilibrium.

\begin{example}
\label{ex:1}
    Let us consider the case in which we have 5 agents and need to place two facilities with $k_1=k_2=2$.
    Let $\vec x:=(0,0.3,0.4,0.5,0.9)\in [0,1]^5$ be the vector containing the agents' positions.
    If $\vec y=(0.3,0.5)$, both $\gamma_1=(1,1,2,2,2)$ and $\gamma_2=(1,1,1,2,2)$ are pure NE of the FCFS game.
    However, the SW of the FCFS game depends on the specific NE, indeed $SW_{\gamma_1}(\vec x,\vec y)=3.6> 3.5=SW_{\gamma_2}(\vec x,\vec y)$.
    Moreover, the utility the agent located at $0.9$ is zero or $0.6$ depending on the equilibrium.
\end{example}

\begin{figure}[t]
    \centering
    \begin{subfigure}{\textwidth}
        \centering
        \includegraphics[width=\textwidth]{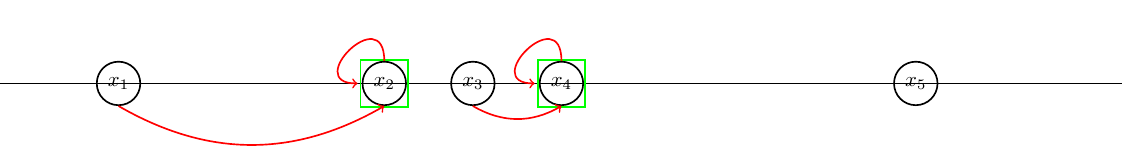}
        \caption{The Nash Equilibrium $\gamma_1$ from Example \ref{ex:1}.}
        \label{fig:first}
    \end{subfigure}
    
    \vspace{0.2cm} 

    \begin{subfigure}{\textwidth}
        \centering
        \includegraphics[width=\textwidth]{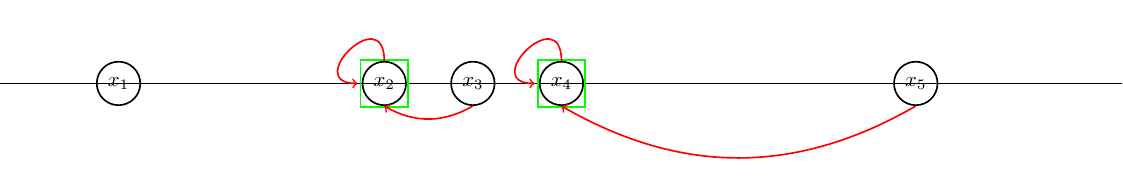}
        \caption{The Nash Equilibrium $\gamma_2$ from Example \ref{ex:1}.}
        \label{fig:second}
    \end{subfigure}
    
    \caption{The two Nash Equilibria described in Example \ref{ex:1}. The circles represents the agents, the green squares the facilities and the red arrows the strategies played by the agents getting a non-null utility.}
    \label{fig:overall}
\end{figure}



%
Lastly, we define the optimal SW achievable on the instance $\vec x$ as 
\begin{equation}
    \label{eq:optSW}
    SW_{opt}(\vec x)=\sup_{\vec y\in[0,1]^m}\sup_{\gamma\in NE(\vec x,\vec y)}SW_{\gamma}(\vec x,\vec y).
\end{equation}

\subsection{Mechanism Design Aspects of the problem}
A key property that a mechanism $M$ must possess is truthfulness, which ensures that no agent can increase their utility by misreporting their position.
As shown in Example \ref{ex:1}, the utility of an agent depends on the other agents' strategies.
For this reason, we employ a stronger notion of truthfulness that keeps track of the different strategies that agents can play in the FCFS game.

\begin{definition}
\label{def:truthfulness}
    A mechanism $M$ is \textit{absolutely truthful} if no agent increases its utility by misreporting, regardless of the strategies played by other agents.
    More formally, for every agent $i\in [n]$, $\vec x\in[0,1]^n$, and $\vec s_{-i}\in [m]^{n-1}$, we have 
    \[
    \max_{s_i\in[m]}u_i(\vec x,M(\vec x); s_i,\vec s_{-i})\ge \max_{s_i'\in[m]}u_i(\vec x ,M(\vec x');s_i',\vec s_{-i}),
    \]
    where \begin{enumerate*}[label=(\roman*)]
       \item $\vec x'=(x_i',\vec x_{-i})$ for every $x_i'\in[0,1]$,
       \item $\vec x_{-i}$ and $\vec s_{-i}$, are the vectors containing the positions and strategies of the other $n-1$ agents, respectively.
    \end{enumerate*}
\end{definition}
In general, the FCFS game induced by the output of a mechanism has multiple NE, hence the SW of the mechanism is not unique.
For this reason, we introduce the notion of Equilibrium Stable mechanism.

\begin{definition}
    An absolutely truthful mechanism $M$ is said Equilibrium Stable (ES) if, for every $\vec x$, we have that
    \[
        SW(\vec x,M(\vec x);\gamma)=SW(\vec x,M(\vec x);\gamma')
    \]
    for every Nash Equilibria $\gamma,\gamma'\in NE(\vec x,M(\vec x))$.
\end{definition}

An absolutely truthful mechanism is not necessarily ES.
For example, consider the mechanism that places the facilities at $(0.3,0.5)$ regardless of the agents' reports.
No agent can manipulate the outcome of the mechanism by changing their reports, however, going back to Example \ref{ex:1}, we have that the mechanism induces two NE with different SW on the instance $(0,0.3,0.4,0.5,0.9)$.
Finally, since we are considering the utilities of the agents, we define the approximation ratio of an ES mechanism as the worst-case ratio of the optimal SW and the SW achieved by the mechanism.

\begin{definition}
    Let $M$ be an absolutely truthful and ES mechanism.
    We define the approximation ratio of $M$ as
    \[
    ar(M)=\sup_{\vec x \in [0,1]^n}\frac{SW_{opt}(\vec x)}{SW_{M}(\vec x)},
    \]
    where $SW_{M}$ is the SW value achieved by $M$ on any of the NE in $NE(\vec x, M(\vec x))$ and $SW_{opt}$ is defined in \eqref{eq:optSW}.
\end{definition}

\section{Absolutely Truthful and Equilibrium Stable Mechanisms to place more than one facility}

In this section, we study the class of percentile mechanisms and characterize under which conditions a percentile mechanism is absolutely truthful and ES.

\begin{definition}[Percentile Mechanism, \cite{sui2013analysis}]
    Given a percentile vector $\vec v\in[0,1]^m$, the routine of the percentile vector associated with $\vec v$, namely $\PMp$, is as follows:
\begin{enumerate*}[label=(\roman*)]
    \item The mechanism designer collects all the reports of the agents $\{x_1,x_2,\dots,x_n\}$ and sorts them in non-decreasing order, so that $x_i\le x_{i+1}$.
    \item The mechanism then places the $m$ facilities at the positions $y_j=x_{i_j}$, where $i_j=\floor{(n-1)v_j}+1$ for every $j\in[m]$.
\end{enumerate*}
\end{definition}

All percentile mechanisms are absolutely truthful.

\begin{theorem}
\label{thm:perctruth}
   $\PMp$ is absolutely truthful for every $\vec v$.
\end{theorem}

\begin{proof}
Toward a contradiction, let $\PMp$ be a percentile vector such that, on instance $\vec x\in[0,1]^m$, the agent whose real position is $x_i$ can manipulate by reporting $x_i'\in[0,1]$.
Let $\vec y$ be the output of $\PMp$ in the truthful input and let $\vec y\,'$ be the output of $\PMp$ after the agent manipulation.
If $x_i'\le x_i$, we have that the position of the facilities that $\PMp$ places on the left of $x_i$ move further to the left of $x_i$.
Each facility that was placed by $\PMp$ on the right of $x_i$ does not change position.
Thus it holds $|x_i-y_j|\le|x_i-y_j'|$ for every $j\in[m]$.
Finally, let $\vec s_{-i}\in[m]^{n-1}$ be a vector containing the strategies of the other agents and let us define with $F_i(\vec z)\subset [m]$ the set of strategies that give a non-null utility to the agent at $x_i$ when the facilities are located at $\vec z$.
Since $|x_i-y_j|\le|x_i-y_j'|$ for every $j\in[m]$, we have that $F_i(\vec y')\subset F_i(\vec y)$.
To conclude, notice that for every $s_i\in F_i(\vec y')$ the utility of the manipulative agent decreases, as all the distances from the facilities have increased after the manipulation, which concludes the proof.
\end{proof}

Unfortunately, not every percentile mechanism is ES: consider the situation represented in Example \ref{ex:1}: the position of the facilities are the output of the $\PMp$ with $\vec v=(0.25,0.75)$, however, different NE induce different SW values.
In what follows, we characterize the set of percentile mechanisms that are ES and study their approximation ratio.
Owing to this characterization, we identify the ES percentile mechanism with the lowest approximation ratio.
For the sake of simplicity, we start our discussion from the case in which $m=2$.

\subsection{Mechanisms to Place Two Facilities}
\label{sec:twofacilities}

First, we study the case in which we place two facilities.
We denote with $k_1$ and $k_2$ the capacities of the two facilities and assume that $k_1+k_2< n$.
Without loss of generality, let $k_1 \ge k_2$.
We show that each ES percentile mechanism belong to one of the following three categories: 
\begin{enumerate*}[label=(\roman*)]
\item the All-In-One (AIO) mechanisms, that place all the facilities at the same spot, that is $y_1=y_2$;
\item the Side-By-Side (SBS) mechanisms, that places the facilities at two adjacent positions, that is $y_1=x_i$ and $y_2=x_{i+1}$ where $i\in [n-1]$; and
\item the Wide-Gap (WG) mechanism, which places the facilities at two agents positions that are well-separated, i.e. $y_1=x_i$ and $y_2=x_j$ where $i+1<j$.
\end{enumerate*}
Notice that each percentile mechanism belongs to exactly one of these three sub-classes, making the partition complete.
We start from the Wide-Gap mechanisms, as they are the more interesting class of ES percentile mechanisms.

\subsubsection{Wide-Gap mechanisms}

First, we show that for every $k_1$ and $k_2$, there exists at least a WG percentile mechanism that is ES.
Moreover, for every $k_1$ and $k_2$, we compute the approximation ratio of every ES percentile mechanism and characterize the mechanism achieving the lowest approximation ratio.
Notice that by definition of WG mechanisms, and without loss of generality, we assume that $v_1<v_2$.
Thus the facility with the highest capacity is always located on the left of the less capacious facility.
The case in which $v_2<v_1$ is symmetric and thus it does not need to be considered.

\begin{theorem}
\label{thm:characterization2facilities}
    Let $\vec v=(v_1,v_2)\in[0,1]^2$ be a percentile vector and let $\PMp$ be its associated percentile mechanism.
    Let $k_1,k_2\in \mathbb{N}$ and let $n\in\mathbb{N}$ be such that $k_1+k_2 < n$ and $\floor{v_2(n-1)}-\floor{v_1(n-1)}>1$.
    Then, we have that $\PMp$ is ES if and only if 
    \begin{equation}
        \label{eq:v_feasible}
        \floor{v_2(n-1)}-\floor{v_1(n-1)}\ge k_2+k_1-1.
    \end{equation}
\end{theorem}

\begin{proof}
    First, we show that condition \eqref{eq:v_feasible} is sufficient to make a percentile mechanism ES.
    If $\vec v$ satisfies \eqref{eq:v_feasible}, there are always $k_1+k_2$ agents such that $y_1\le x_i\le y_2$, where $y_1$ and $y_2$ are the two facility locations returned by the mechanism.
    Notice that if $y_1=y_2$, then the two facilities share the position with $k_1+k_2$ agents, hence the SW of any NE is equal to $k_1+k_2$.
    Assume now that $y_1<y_2$.
    Let us then define $r_j$ the minimal values such that $B_{r_j}(y_j)\cap\{x_i\}_{i\in[n]}$ has cardinality larger or equal to $k_j$\footnote{Here $B_r(y)$ denotes the ball centered in $y$ of radius $r$.}.
    Since there are at least $k_1+k_2$ agents in $[y_1,y_2]$, we have that $r_1+r_2\le |y_1-y_2|$.
    According to every NE, agents that do not belong to either $B_{r_1}(y_1)$ or $B_{r_2}(y_2)$ have utility equal to $0$.
    Indeed, if $\vec s\in NE(\vec x, \PMp(\vec c))$ is such that $x_i\notin B_{r_1}(y_1)\cap B_{r_2}(y_2)$ gets accommodated by $y_1$, we would have that at least one agent $x_j\in B_{r_1}(y_1)$ has null utility, hence $s_j=2$.
    However, if agent $x_j$ can increase its utility by changing its strategy to $s_j=1$, hence $\vec s\notin NE(\vec x, \PMp(\vec x))$, which is a contradiction.
    By the same argument, we infer that every agent $x_i\in(y_j-r_j,y_j+r_j)$ attains utility $1-|x_i-y_j|$ according to every NE.
    Finally, we observe that the set of agents such that $|x_i-y_j|=r_j$ that have non null utility may change according to the specific NE, but the total utility of these agents does not change.
    Thus condition \eqref{eq:v_feasible} is sufficient to ensure $\PMp$ is ES.
    Lastly, we show that the condition is necessary.
    Let us assume that $\vec v$ does not satisfy the condition \eqref{eq:v_feasible}.
    For the sake of simplicity, let us denote with $i_1=\floor{v_1(n-1)}+1$ and $i_2=\floor{v_2(n-1)}+1$.
    Let us consider the following instance $x_1=\dots=x_{i_1-1}=0$, $x_{i_1}=0.4$, $x_{i_1+1}=\dots=x_{i_2-1}=0.5$, $x_{i_2}=0.6$ and $x_j=0.9$ for all the other indexes $j>i_2$.
    Notice that, since $\floor{v_2(n-1)}-\floor{v_1(n-1)}>1$, there is at least one agent located at $0.5$.
    Following the same argument used in Example \ref{ex:1}, we can show that, depending on the strategy played by the agents at $0.5$, the Social Welfare of the mechanism changes.
\end{proof}

Next, we characterize the approximation ratio of every WG percentile mechanism.
Given a percentile vector $\vec v\in[0,1]^2$ that satisfies condition \eqref{eq:v_feasible}, we denote with $i_1=\floor{v_1(n-1)}+1$ and $i_2=\floor{v_2(n-1)}+1$.
Therefore that the mechanism places the facility with capacity $k_1$ at $x_{i_1}$ and the facility with capacity $k_2$ at $x_{i_2}$.

\begin{theorem}
\label{thm:i_1+i_2}
    Given $n$, $k_1$, and $k_2$, let $\PMp$ an ES percentile mechanism.
    Then, if $i_1\ge \floor{\frac{k_1+1}{2}}$, we have that
    \begin{equation}
        ar(\PMp)=\frac{k_1+k_2}{\min\{k_1+(n-i_2)+1,\frac{k_1+1}{2}+k_2\}}
    \end{equation}
    If $i_1 < \floor{\frac{k_1+1}{2}}$ and $i_2< n-\floor{\frac{k_2+1}{2}}$, we have
    \[
    ar(\PMp)=\frac{k_1+k_2}{\min\{k_1+(n-i_2)+1,i_1+k_2\}}.
    \]
    Otherwise, we have 
    \[
    ar(\PMp)=\frac{k_1+k_2}{\min\{k_1+\frac{k_2+1}{2},k_2+i_1\}}.
    \]
\end{theorem}

\begin{proof}
Our argument is as follows: we show that the worst instance for the mechanism is either \begin{enumerate*}[label=(\roman*)]
    \item \label{thm4case1} $x_i=0$ if $i\in\{1,\dots,i_1-1\}$, $x_{i_1}=\frac{1}{2}$, and $x_i=1$ otherwise, or
    \item \label{thm4case2} $x_i=0$ if $i\in\{1,\dots,i_2-1\}$ and $x_i=1$ otherwise.
\end{enumerate*}
Notice that in both cases, the optimal SW is equal to $k_1+k_2$, which is the maximum SW attainable.
Let us show that the worst-case instance has the form described in \ref{thm4case1} or \ref{thm4case2}.
Owing to Theorem \ref{thm:characterization2facilities}, we have that there are at least $k_1+k_2$ agents in the interval $[y_1,y_2]$, hence the agents that are accommodated by the facility at $y_j$ are, up to ties, the $k_j$ agents that are closer to $y_j$.
Since $i_1\ge \floor{\frac{k_1+1}{2}}$, the total utility of the agent accommodated by the facility at $y_1$ is minimized when all the agents accommodated by $y_1$ are all at the same distance from $y_1$, that is $|x_{i_1-1}-y_1|=|x_{i_1+1}-y_1|$.
Given $\lambda\in[0,\frac{1}{2}]$, let us consider the following instance: $x_1=\dots=x_{i_1-1}=0$, $x_{i_1}=\lambda$, $x_{i_1+1}=\dots=x_{i_2-1}=2\lambda$.
Let us now consider the facility located at $y_2$.
By the same argument, we have that if $n-i_2$ is larger than $\floor{\frac{k_2+1}{2}}$, then, for every $\lambda$, the position $y_2$ that minimizes the utility is $\frac{1}{2}+\lambda$.
In this case, the cost of the mechanism is $2+(1-\lambda)(k_1-1)+(\frac{1}{2}+\lambda)(k_2-1)$.
Since $k_1\ge k_2$, we have that the minimum SW is achieved when $\lambda=\frac{1}{2}$, thus all the agents on the right of $x_{i_1}$ are located at $1$, all the agents on the left are located at $0$ and $x_{i_1}=\frac{1}{2}$.
In this case $SW_{\PMp}(\vec x)=\frac{k_1+1}{2}+k_2$, while $SW_{opt}(\vec x)=k_1+k_2$, which is the maximum utility achievable and is attained by placing both facilities at $1$.
Consider now the case $n-i_2<\floor{\frac{k_2+1}{2}}$.
In this case, for every $\lambda$, the worst position for $y_2$ is $1$, hence the SW of the mechanism is $n-i_2+2+(1-\lambda)(k_1-1)+2\lambda(k_2-n+i_2-1)$, therefore the SW is minimized when $\lambda=0$ or $\lambda=\frac{1}{2}$.
In particular, we have that worst Social Welfare attainable by the mechanism is
\[
    \min\{n-i_2+k_1+1),\frac{1}{2}(k_1+1)+k_2)\}.
\]
Since $SW_{opt}(\vec x)=k_1+k_2$, we conclude the first part of the proof.
We now consider the case in which $i_1<\floor{\frac{k_1+1}{2}}$.
First, we consider the case in which $i_2<n-\floor{\frac{k_2+1}{2}}$.
By the same argument used to prove the case in which $i_1\ge \floor{\frac{k_1+1}{2}}$, we have that the worst-case instance in this case is
\[
x_i=\begin{cases}
    x_i=0\quad\quad\text{if}\quad i=1,\dots,i_1,\\
    x_i=\lambda\quad\quad\text{if}\quad i=i_1+1,\dots,i_2-1,\\
    x_i=1\quad\quad\text{otherwise.}
\end{cases}
\]
for some $\lambda\in[0,1]$, since the SW of the mechanism is minimized when the $i_1$-th and $i_2$-th agents are at the extremes of the interval.
For any value of $\lambda$, the SW of the mechanism is
\[
SW(\vec x)=i_1+(n-i_2)+(1-\lambda) (k_1-i_1) + \lambda (k_2-(n-i_2)).
\]
Since $SW(\vec x)$ is linear in $\lambda$, we have that the minimum is achieved at either $\lambda=0$ or $\lambda=1$.
Thus the minimal SW achievable is
\[
\min\{k_1+(n-i_2),k_2+i_1\}.
\]
Since in both cases we have that the optimal SW is $k_1+k_2$, we conclude the thesis.
Lastly, we consider the case in which $n-i_2\ge \floor{\frac{k_2+1}{2}}$ and $i_1<\floor{\frac{k_1+1}{2}}$.
In this case, the worst-case instance places the first $i_1$ agents at $0$, therefore the instances we need to consider are 
%
\[
x_i=\begin{cases}
    x_i=0\quad\quad\text{if}\quad i=1,\dots,i_1,\\
    x_i=\lambda\quad\quad\text{if}\quad i=i_1+1,\dots,i_2-1,\\
    x_{i_2}=\frac{\lambda+1}{2}\\
    x_i=1\quad\quad\text{otherwise.}
\end{cases}
\]
The SW induced by the mechanism is then
\[
SW(\vec x)=i_1+ 1 +(1-\lambda)(k_1-i_1)+\frac{1+\lambda}{2}(k_2-1).
\]
Again, since the SW is linear in $\lambda$, we have that the minimium is attained at either $\lambda=0$ or $\lambda=1$.
Then the minimum SW achievable by the mechanism is
\[
\min\Big\{k_1+\frac{(k_2+1)}{2},k_2+i_1\Big\}.
\]
To conclude notice that in both cases, the SW attained by the optimal solution is $k_1+k_2$.
\end{proof}

Consequentially, we characterize the best ES percentile mechanisms given any $2$-dimensional vector $\vec k$.
In particular, we show that the approximation of the best percentile mechanism decreases as $\Delta:=n-(k_1+k_2)$ increases.

\begin{theorem}
\label{thm:bestPMPmechanism}
Given $n$ and $\vec k$, let us define $\Delta=n-(k_1+k_2)$, then we have that the ES WG percentile mechanism that achieves the lowest approximation ratio is induced by the percentile vector $\vec v=(\frac{i_1}{n},\frac{i_2}{n})$, where $i_1$ and $i_2$ are as follows
\begin{enumerate}[label=(\roman*)]
    \item \label{thm5case1} $i_1=\ceil{\frac{k_1}{2}}$ and $i_2=n-\floor{\frac{k_2}{2}}$ if $\Delta\ge \ceil{\frac{k_1+k_2}{2}}$, in which case $ar(\PMp)=\frac{k_1+k_2}{\frac{k_1+1}{2}+k_2}$,
    \item \label{thm5case2} $i_1=k_1-k_2+\alpha$ and $i_2=n-\alpha$, where $\alpha=\ceil{\frac{\Delta-(k_1-k_2)}{2}}$, if $k_1-k_2\le\Delta\le \floor{\frac{k_1+k_2}{2}}+1$, in which case $ar(\PMp)=\frac{k_1+k_2}{i_1+k_2}$, and
    \item \label{thm5case3} $i_1=\Delta+1$ and $i_2=n$ otherwise, in which case $ar(\PMp)=\frac{k_1+k_2}{\Delta+k_2+1}$.
\end{enumerate}
\end{theorem}

\begin{proof}
When $\Delta\ge\ceil{\frac{k_1+k_2}{2}}$, the indexes $i_1=\ceil{\frac{k_1}{2}}$ and $i_2=n-\floor{\frac{k_2}{2}}$ are well defined.
Owing to Theorem \ref{thm:characterization2facilities} and by definition of $\Delta$, we have that $\PMp$ is ES.
Finally, from Theorem \ref{thm:i_1+i_2}, we infer that 
\[
ar(\PMp)=\frac{k_1+k_2}{\frac{k_1+1}{2}+k_2},
\]
which is the smallest approximation ratio achievable by an ES WG percentile mechanism.
%

%

%
Let us consider the case \ref{thm5case2}, that is $k_1-k_2\le\Delta\le \floor{\frac{k_1+k_2}{2}}+1$.
Owing to Theorem \ref{thm:i_1+i_2}, we retrieve the best values $i_1$ and $i_2$ by maximizing the quantity
\[
    \min\{k_1+(n-i_2),i_1+k_2\}.
\]
Thus, we look for $i_1$ and $i_2$ such that
\[
    k_1+(n-i_2)=i_1+k_2,
\]
subject to the constraint 
\[
n-i_2+i_1=\Delta,
\]
since, owing to Theorem \ref{eq:v_feasible}, $k_1+k_2$ agents must lay between $x_{i_1}$ and $x_{i_2}$.
By a simple computation, we have that
\[
n-i_2=\frac{k_2-k_1+\Delta}{2},
\]
thus $i_1=\frac{\Delta-(k_2-k_1)}{2}=k_1-k_2+\frac{\Delta-(k_2-k_1)}{2}$ and $i_2=n-\frac{k_2-k_1+\Delta}{2}$, which concludes the proof of case \ref{thm5case2}.
Lastly, we consider case \ref{thm5case3}.
In this case, we have that $\Delta<k_1-k_2$, thus we have
\[
k_2+i_1-k_1-(n-i_2)=i_2-n+i_1+k_2-k_1\le \Delta+k_2-k_1<0,
\]
since $i_2-n+i_1<n-i_2+i_1\le \Delta$.
Thus the minimum SW attainable by the mechanism is $i_1+k_2$, therefore, to maximize the minimum achievable SW, we need to set $i_1=\Delta$ and $i_2=n$, which concludes the proof.
\end{proof}

Notice that the lowest approximation ratio is achieved when $\Delta\ge \ceil{\frac{k_1+k_2}{2}}$.
Moreover, notice that the smaller the gap between $k_1$ and $k_2$, that is $k_1-k_2$, the lower the approximation ratio of the best percentile mechanism.
In particular, the lowest approximation ratio is attained when $k_1=k_2$ and $n\ge 3k$, in which case there exists a percentile mechanism whose approximation ratio is $\frac{4}{3+\frac{1}{k}}\sim\frac{4}{3}$.

\subsubsection{The All-In-One and Side-By-Side mechanisms}

We conclude our study by considering the AIO and the SBS mechanisms.
Unlike the WG mechanisms, each AIO and SBS is ES.
We consider the AIO mechanisms first.
Since all the AIO mechanisms place the facilities at the same position, every Nash Equilibrium induced by the facility placement achieves the same Social Welfare.
Moreover, it is easy to see that the best AIO mechanism places the both the facilities at the median agent.
Indeed, we have the following.

\begin{theorem}
\label{thm:ar_med}
    The best AIO mechanism places all the facilities at the median agent.
    Moreover, denoted with $M$ the mechanism that places both facility at the median agent,  we have that
    \begin{equation}
    \label{eq:median_AIO_ar}
    ar(M)=\begin{cases}
        \frac{2k_2+2\floor{\frac{n}{2}}+1}{k_1+k_2+1}\quad\quad \text{if}\;\; k_1\ge \ceil{\frac{n}{2}}\\
        \\
        \frac{2(k_2+k_1)}{k_1+k_2+1}\quad\quad\quad\:\; \text{otherwise}.
    \end{cases}
    \end{equation}
\end{theorem}

\begin{proof}
    We prove this theorem as it follows: first, we compute the approximation ratio of a generic AIO mechanism and then show that the approximation ratio is minimized when the AIO mechanism places both facility at a median agent.
    Let then $M$ be an AIO mechanism such that $y_1=y_2=x_{r}$ where $r\in[n]$.
    Without loss of generality, we assume $r\le \floor{\frac{n+1}{2}}$.
    We notice that, given an instance $\vec x$, it is possible to increase the optimal Social Welfare and decrease the Social Welfare of the mechanism by moving all the agents that are not the $r$-th agent to either $0$ or $1$.
    We therefore restrict our attention to instances such that $x_1=\dots=x_{r-1}=0$, $x_r=\lambda\in[0,1]$, and $x_j=1$ otherwise.
    Since we have that $r\le \floor{\frac{n}{2}}$, we have that the worst-case instance happens when $\lambda\in[0,\frac{1}{2}]$, thus the Social Welfare achieved by the mechanism can be computed explicitly
    \[
        SW_M(\vec x)=1+(1-\lambda)\min\{r-1,k_1+k_2-1\}+\lambda(k_1+k_2-\min\{k_1+k_2,r\}).
    \]
    To complete the proof we need to compute the optimal SW achievable on these instances.
    Since $r\le \floor{\frac{n}{2}}$ and $k_1\ge k_2$, we have that the optimal solution either places the two facilities at $1$, in which case the SW is equal to
    \[
        SW_{opt,1}=\min\{n-r+\lambda,k_1+k_2\},
    \]
    or places the facility with capacity $k_2$ at $0$ and the facility with capacity $k_1$ at $1$, in which case the optimal SW is equal to
    \[
        SW_{opt,2}=\min\{r-\lambda,k_2\}+\min\{n-r+\lambda,k_1\}.
    \]
    So that $SW_{opt}=\max\{SW_{opt,1},SW_{opt,2}\}$ and hence
    \[
        ar(M)=\max_{\lambda\in[0,\frac{1}{2}]}\Bigg\{\frac{\max\big\{\min\{n-r+\lambda,k_1+k_2\},\min\{r-\lambda,k_2\}+\min\{n-r+\lambda,k_1\}\big\}}{1+(1-\lambda)\min\{r-1,k_1+k_2-1\}+\lambda(k_1+k_2-\min\{k_1+k_2,r\})}\Bigg\}.
    \]
    We notice that this quantity is minimized when $r=\ceil{\frac{n}{2}}$, which proves that the best AIO mechanism is the one placing all the facilities at the median agent.
    Let us now compute the approximation ratio of the median mechanism, that is $r=\ceil{\frac{n}{2}}$.
    Moreover, since $SW_{opt}\ge SW_M$, $k_1\ge k_1$, and both $SW_M$ and $SW_{opt}$ are piecewise linear functions with respect to $\lambda$, we have that the maximum defining $ar(M)$ is attained when wither $\lambda=0$ or $\frac{1}{2}$.
    We then conclude that
    \[
        ar(M)=\max_{\lambda\in\{0,\frac{1}{2}\}}\Bigg\{\frac{\max\big\{\min\{\floor{\frac{n}{2}}+\lambda,k_1+k_2\},k_2+\min\{\floor{\frac{n}{2}}+\lambda,k_1\}\big\}}{\lambda+(1-\lambda)\min\{\floor{\frac{n}{2}},k_1+k_2\}+\lambda(k_1+k_2-\min\{k_1+k_2,\ceil{\frac{n}{2}}\})}\Bigg\},
    \]
    where we used the fact that $k_2\le \ceil{\frac{n}{2}}$.
    Which allows us to conclude \eqref{eq:median_AIO_ar} and thus the proof.
    %
%
%
\end{proof}

We then consider the SBS mechanisms.
As for the AIO mechanisms, every SBS mechanism is ES.

\begin{theorem}
\label{thm:SBS_ES}
   Each SBS mechanism is ES. 
\end{theorem}

\begin{proof}
    Let $PM$ be an SBS percentile mechanism and let $i$ be the parameter such that $y_1=x_i$ and $y_2=x_{i+1}$.
    Without loss of generality, we assume that the facility at $y_1$ has capacity $k_1$ and that the facility at $y_2$ has capacity $k_2$.
    Let $\vec x$ be the vector containing the reports of $n$ agents, we now show that every NE equilibrium induced by placing the facilities at $y_1=x_i$ and $y_2=x_{i+1}$ achieves the same SW.
    Notice that if $y_1=y_2$, there is nothing to prove, we then consider the case in which $y_1<y_2$.
    Let us now define
    \begin{equation}
        \label{eq:A1_A2}
            A_1:=\{x_i\le y_1\}\quad\quad \text{and} \quad \quad A_1:=\{x_i\ge y_2\}.
    \end{equation}
    Notice that $A_1\cup A_2=\{x_i\}_{i\in[n]}$ and $A_1 \cap A_2=\emptyset$.
    By definition, every agent in $A_1$ is strictly closer to $y_1$ and every agent in $A_2$ is strictly closer in $y_2$.
    We have three cases:
    \begin{enumerate*}[label=(\roman*)]
        \item \label{thm7case1}$|A_1|\le k_1$ and $|A_2|\le k_2$;
        \item \label{thm7case2}$|A_1|> k_1$ and $|A_2|\le k_2$; and
        \item \label{thm7case3}$|A_1|\le k_1$ and $|A_2|> k_2$.
    \end{enumerate*}

    \textbf{Case \ref{thm7case1}:} In this case, we have that in every NE, the agents getting access to the facilities are, up to breaking ties, the $k_1$ rightmost agents in $A_1$ and the $k_2$ leftmost agents in $A_2$. 
    In particular the SW remains constant across all the possible Nash Equilibria.
    
    \textbf{Case \ref{thm7case2}:} In this case, we have that in every NE, the agents getting access to the facilities are, up to breaking ties, the $a_1:=|A_1|$ agents in $A_1$ and the $k_2+k_1-a_1$ leftmost agents in $A_2$.
    Notice that, up to ties, the $k_2$ leftmost agents in $A_2$ are served by the facility at $y_2$, while the other $k_1-a_2$ are served by the facility $y_1$.
    In particular the SW remains constant across all the possible Nash Equilibria.
     
    \textbf{Case \ref{thm7case3}:} This case follows from an argument similar to the one used to handle Case \ref{thm7case2}. Thus we conclude the proof.
\end{proof}

Lastly, we notice that the approximation ratio of any SBS mechanism can be obtained by following the same argument used to prove Theorem \ref{thm:ar_med}.
Indeed, it is easy to see that the worst-case instance of the SBS places the two facilities at the same position, i.e. $y_1=y_2$, as otherwise it would be possible to increase the ratio by overlapping the two facilities.
For this reason, the worst-case guarantees of the SBS mechanisms are similar to the worst-case guarantees of the AIO mechanisms, hence the SBS mechanism that achieves the lowest approximation ratio is the one placing facilities at $x_{\floor{\frac{n}{2}}}$ and $x_{\floor{\frac{n+1}{2}}}$.

\subsection{Beyond two facilities}

We now extend our study to the case in which we want to place $m>2$ facilities.
For the sake of simplicity, we consider $m$ facilities that have the same capacity $k$.
The techniques used in this section can be easily extended to the case in which facilities have different capacities. 
Notice that, in this case, any percentile mechanism needs to be endowed with a permutation to specify how to assign the different capacities to the positions returned by the mechanism \cite{aziz2020facility,auricchio2023extended}.
First, we extend Theorem \ref{thm:characterization2facilities} to this framework.

\begin{theorem}
\label{thm:ESPMmmore2}
    Let $k$ be the capacity of $m$ facilities.
    Moreover, let $\vec v$ be a percentile vector such that $v_1<v_2<\dots<v_m$ so that $\vec v$ does not possess two equal entries and let $\PMp$ be its associated percentile mechanism.
    Assume that $\floor{v_{j+1}(n-1)}-\floor{v_j(n-1)}>1$ for every $j\in [m-1]$.
    Then $\PMp$ is ES if and only if the following system of inequalities is satisfied
    \begin{eqnarray}
    \label{eq:Esconditionmmore2}
        \begin{cases}
            &\floor{v_2(n-1)}-\floor{v_1(n-1)}\ge 2k-1\\
            &\dots\\
            &\floor{v_{m}(n-1)}-\floor{v_{m-1}(n-1)}\ge 2k-1
        \end{cases}.
    \end{eqnarray}
\end{theorem}

\begin{proof}
    The proof follows by the same argument used to prove Theorem \ref{thm:characterization2facilities}.
    Indeed, by condition \eqref{eq:Esconditionmmore2} for every $j\in[m]$ we have that at least $2k$ agents are located between $y_j$ and $y_{j+1}$, thus the Social Welfare generated by the facilities at $y_j$ and $y_{j+1}$ does not depend on the specific Nash Equilibrium.
    To conclude the proof, it suffices to apply this argument to each couple of facilities $(y_j,y_{j+1})$.
\end{proof}

The set of inequalities \eqref{eq:Esconditionmmore2} allows us to characterize the vectors $\vec v$ that induces an ES percentile mechanism $\PMp$ depending on the capacity $k$.
Notice that system \eqref{eq:Esconditionmmore2} does not admit any solution when $k>\frac{n+m}{2m}$ or, equivalently, $n<(2k-1)m$.
Indeed, by summing all the inequalities in \eqref{eq:Esconditionmmore2}, we have that
\[
\floor{v_m(n-1)}-\floor{v_1(n-1)}\ge (2k-1)m.
\]
Since $n\ge \floor{v_m(n-1)}-\floor{v_1(n-1)}$, we must have that $n\ge (2k-1)m$.
Although when $n<(2k-1)m$ it is impossible to define an ES percentile mechanism that places $m$ facilities at $m$ different locations, it is possible to define an ES percentile mechanism that places all the facilities at one or two different locations.
To keep the discussion on track, we first study the case in which system \eqref{eq:Esconditionmmore2} admits a solution and defer the pathological case to a dedicated section.

\subsection{Case \texorpdfstring{$n\ge (2k-1)m$}{Lg}.}

In this case, it is possible to select an ES and absolutely truthful percentile mechanism that places the $m$ facilities at $m$ different positions among the agents' reports.

\begin{theorem}
\label{thm:approximationratiom>2}
    If $k<\frac{n+m}{2m}$, then given an ES $\PMp$, we have
    \[
        ar(\PMp)=\begin{cases}
            \frac{mk}{(m-\frac{1}{2})k+\frac{1}{2}} \quad \quad \quad \text{if} \; i_1,n-i_m\ge\floor{\frac{k+1}{2}}\\
            \frac{mk}{(m-1)k+\min\{i_1,n-i_m\}} \quad \text{otherwise}
        \end{cases} 
    \]
    where $i_1=\floor{v_1(n-1)}+1$ and $i_m=\floor{v_m(n-1)}+1$.
\end{theorem}

\begin{proof}
    The case in which $i_1,n-i_m\ge\floor{\frac{k+1}{2}}$ follows by the same argument adopted in the proof of Theorem \ref{thm:i_1+i_2}.
    Indeed, by definition of the mechanism, the SW of the mechanism is minimized when each facility $y_j=x_{\floor{v_j(n-1)}+1}$ is such that $|y_j-x_{\floor{v_j(n-1)}}|=|y_j-x_{\floor{v_j(n-1)}+2}|$.
    Hence the mechanism achieves the minimal SW when $x_{\floor{v_j(n-1)}+1}=\frac{2j-1}{2m}$ for every $j\in[m]$ and $x_i=\frac{l}{m}$ if $\floor{v_l(n-1)}+1 < i < \floor{v_{l+1}(n-1)}+1$ where $l=0,1,\dots,m$, $v_0=0$ and $v_{m+1}=1$.
    On such instance the SW of the mechanism is $(m-\frac{1}{2})k+\frac{1}{2}$.
    Notice the mechanism achieves the same SW on the instance $\vec x_O$ defined as $(x_O)_i=0$ for every $i\le \floor{v_{1}(n-1)}+1$, and $(x_O)_i=1$ otherwise.
    To conclude, we observe that the optimal SW on instance $\vec x_O$ is $mk$.
    To conclude the proof, we need to consider the case in which either $i_1$ or $n-i_m$ are lower than $\floor{\frac{k+1}{2}}$.
    Since the other case is symmetric, we restrict our analysis to the case in which $i_1\le n-i_2$.
    Again, since $i_1, n-i_m\le\floor{\frac{k+1}{2}}$, we have that the worst-case instance places the first $i_1$ agents at $0$ and the last $n-i_m+1$ at $1$.
    Since every facility has the same capacity, we have that the worst-case instance has the following form
    \[
    x_i=\begin{cases}
        0\quad\quad\quad\quad\quad\quad\;\text{if}\quad i=1,\dots,i_1,\\
        \delta_1\quad\quad\quad\quad\;\;\;\quad\text{if}\quad i=i_1+1,\dots,i_2-1,\\
        \delta_1+\frac{1-\delta_1-\delta_2}{2(m-2)}\quad\;\text{if}\quad i=i_2,\\
        \delta_1+2\frac{1-\delta_1-\delta_2}{2(m-2)}\quad\text{if}\quad i=i_2+1,\dots,i_3-1,\\
        \delta_1+3\frac{1-\delta_1-\delta_2}{2(m-2)}\quad\text{if}\quad i=i_3,\\
        \delta_1+4\frac{1-\delta_1-\delta_2}{2(m-2)}\quad\text{if}\quad i=i_3+1,\dots,i_4-1,\\
        \dots\\
        1-\delta_2\quad\quad\quad\quad\text{if}\quad i=i_{m-1}+1,\dots,i_m-1,\\
        1\quad\quad\quad\quad\quad\quad\text{otherwise}
    \end{cases}
    \]
    where $\delta_1,\delta_2\ge 0$ and such that $\delta_1+\delta_2\le 1$.
    The SW of the mechanism on this instance is
    \begin{align*}
        SW(\vec x)&=i_1+(n-i_2)+m-2+(k-i_1)(1-\delta_1)+\sum_{i=2}^{m-2}\bigg((k-1)\Big(\frac{m-3+\delta_1+\delta_2}{m-2}\Big)\bigg)\\
        &\quad+(k-(n-i_m))(1-\delta_2).
    \end{align*}
    Again, this quantity is linear in $\delta_1$ and $\delta_2$, thus it is minimized when $\delta_1,\delta_2\in\{0,1\}$
    By plugging in the possible combinations, we infer that the minimum is achieved when $\delta_1=1$ and $\delta_2=0$ since $i_1\le n-i_m$.
\end{proof}

In particular, for every given the capacity $k$ and number of facility $m$, it is possible to detect the best possible ES and absolutely truthful percentile mechanism.

\begin{theorem}
\label{thm:conditionm>2}
    Given $k$, $m$, and $n$, let us define $\alpha=\floor{\frac{(n-2k(m-1)+1)}{2}}$.
    The vector $\vec v$ where $v_j=\frac{\alpha+(2k-1)(j-1)}{n}$ for $j\in[m]$ induces the ES percentile mechanism with the lowest approximation ratio.
    In particular, if $n\ge 2km$, the approximation ratio of $\PMp$ is less than $1+\frac{1}{2m-1}$.
\end{theorem}

\begin{proof}
    Owing to Theorem \ref{thm:approximationratiom>2}, the approximation ratio is lower when $\min\{i_1,n-i_m\}$ is maximized, thus when $i_1=n-i_m$.
    Thus the best mechanism places the first and last facility at $x_\ell$ and $x_{n-\ell}$, where $\ell$ is a suitable integer.
    Since $i_1+n-i_m=n-2k(m-1)+1$, we complete the first half of the proof.
    Notice that, if $i_1$ or $i_m$ is less than $\floor{\frac{k+1}{2}}$, then we have that $\min\{i_1,i_m\}\le \floor{\frac{k+1}{2}}$.
    By comparing the Social Welfare computed in the Theorem \ref{thm:approximationratiom>2}, we have that $\Big(m-\frac{1}{2}\Big)k+\frac{1}{2}\ge (m-1)k + \min\{i_1,i_m\}$, indeed
    \[
        \Big(m-\frac{1}{2}\Big)k+\frac{1}{2}-(m-1)k- \min\{i_1,i_m\}\ge \frac{k}{2}+\frac{1}{2}- \floor{\frac{k+1}{2}}\ge 0,
    \]
    thus the approximation ratio of the mechanism is smaller when $i_1,i_m\ge \floor{\frac{k+1}{2}}$.
    Moreover, in this case, the approximation ratio does not depend on the specific $\vec v$, thus any ES percentile mechanism whose $\vec v$ is such that $i_1,i_m\ge \floor{\frac{k+1}{2}}$ achieves the minimum approximation ratio.
    Notice that, by definition, the vector $\vec v$ where $v_j=\frac{\alpha+(2k-1)(j-1)}{n}$ for $j\in[m]$ where $\alpha=\floor{\frac{(n-2k(m-1)+1)}{2}}$ is such that $i_1,i_m\ge \floor{\frac{k+1}{2}}$.
    Moreover, owing to Theorem \ref{eq:v_feasible}, it is also ES, hence it achieves the minimal approximation ratio.
    Lastly, notice that
    \[
        \frac{mk}{(m-\frac{1}{2})k+\frac{1}{2}}\le \frac{mk}{(m-\frac{1}{2})k}=\frac{(m-\frac{1}{2})k+\frac{k}{2}}{(m-\frac{1}{2})k}=1+\frac{1}{2m-1},
    \]
    which concludes the proof.
\end{proof}

Notice that, if $n\ge 2km$, the approximation ratio of the best percentile mechanism decreases as the number of facilities increases.
Noticeably, when $m$ goes to infinity, the approximation ratio goes to $1$.

\subsection{Case \texorpdfstring{$n < (2k-1)m$}{Lg}.}

We now consider the case in which the number of agent is too small and thus Theorem \ref{thm:ESPMmmore2} does not hold.
In this case, it is possible to circumvent Theorem \ref{thm:ESPMmmore2}, by considering an percentile mechanism that places all the facilities at either one or two locations, that is the percentile mechanisms whose associated vector $\vec v$ has at most two different entries.
When more than one facility is placed at the same location, we considered them as a unique facility whose capacity is the sum of all the facilities placed at the common location.
Notice that it is sufficient to consider the case in which facilities are divided between at most two points.
Indeed, splitting the facilities between three or more positions leads to another unfeasible system of the form \eqref{eq:Esconditionmmore2}.
Owing to the results of Section \ref{sec:twofacilities}, we have that there are only three ways in which the facilities can be grouped.

\subsubsection{The All-aside mechanisms}
When the mechanism places the facilities at two different locations, we can use the results proposed in Section \ref{sec:twofacilities}.
Indeed, owing to Theorem \ref{thm:bestPMPmechanism}, we know that the approximation ratio becomes lower as the difference in capacity between facilities is smaller.
For this reason, we consider a mechanism that splits the facilities as evenly as possible.

\begin{mechanism}[\texttt{All-aside} mechanism]
    Let $k$ be the capacity of $m$ facilities and let $a,b\in \mathbb{N}$ be such that $a+2mk\le b\le n$. 
    Given in input a vector $\vec x\in[0,1]^n$, the \texttt{All-aside} mechanism associated with $a$ and $b$, namely $AS_{a,b}$, places $\ceil{\frac{m}{2}}$ facilities at $x_a$ and $\floor{\frac{m}{2}}$ facilities at $x_b$.
\end{mechanism}

Owing to Theorem \ref{thm:characterization2facilities}, the \texttt{All-aside} mechanism is absolutely truthful and ES.
Moreover, we can extend Theorem \ref{thm:i_1+i_2} to this case.

\begin{theorem}
\label{thm:allasidear}
    The approximation ratio of every $AS_{a,b}$ is determined by Theorem \ref{thm:i_1+i_2} by setting $k_1=\ceil{\frac{m}{2}}k$, $k_2=\floor{\frac{m}{2}}k$, $i_1=a$, and $i_2=b$.
\end{theorem}

\begin{proof}
    It follows directly from Theorem \ref{thm:i_1+i_2}.
    Indeed, it suffices to prove that even if we have $m$ facilities to locate, the optimal SW we can obtain by locating $m$ facilities with capacity $l$ is the same as locating two facilities with capacity $\ceil{\frac{m}{2}}k$ and $\floor{\frac{m}{2}}k$.
    Since the worst-case instance of any $\PMp$ with $\vec v\in[0,1]^{2}$ places $i_1$ agents $0$ and the others at $1$, the optimal SW remains  $mk$ even though we locate $m$ facilities separately.
\end{proof}

\subsubsection{The AIO and SBS mechanisms}
Lastly, we consider the case in which the mechanism places all the facilities at one place, hence the percentile vector $\vec v=(v,v,\dots,v)$ for a $v\in[0,1]$.
In this case, every $\vec v=(v,v,\dots,v)$ induces an absolutely truthful and ES mechanism.
By the same argument used to prove Theorem \ref{thm:ar_med}, the best percentile mechanism places all the facilities at the median agent, that is $\vec m=(0.5,0.5,\dots,0.5)$.
In our case, however, the approximation ratio guarantees are worse than the one presented in \cite{aziz2020capacity}.
Indeed, since in our case the capacity can be split at $m$ different locations, the optimal solution has a further degree of freedom that heightens the approximation ratio of the mechanism.

\begin{theorem}
\label{thm:allinone}
    Let $k>1$ be the capacity of the facilities and fix $\vec v=(0.5,\dots,0.5)$.
    If $n\le (m+1)k$, we have that
    \[
        ar(\PMp)=\frac{2(m-1)k+(n-(m-1)k)+1}{mk+1}.
    \]
    Otherwise, $ar(\PMp)=\frac{2mk}{mk+1}= 2-\frac{2}{mk+1}$.
\end{theorem}

\begin{proof}
    By definition, we have that for every input $\vec x\in[0,1]^n$ the facility is placed at $x_{\floor{\frac{n+1}{2}}}$.
    The number of agents on the left of $y_1$ and the number of agents on the right of $y_1$ is the same, hence the SW of the mechanism is minimized when $x_i=0$ when $i<\floor{\frac{n+1}{2}}$, $x_{\floor{\frac{n+1}{2}}}=\frac{1}{2}$, and $x_i=1$ otherwise.
    The SW of the mechanism is $\frac{mk+1}{2}$.
    If $n\le (m+1)k$, the optimal SW on the instance is $(m-1)k+\frac{n-(m-1)k}{2}+\frac{1}{2}$.
    Indeed, we can locate $m-1$ facilities at either $0$ or $1$ that only accommodate the agents at $0$ and $1$.
    The total combined utility of the agents accommodated by these $m-1$ facilities is $(m-1)k$.
    Since the agents are divided evenly among $0$ and $1$, the maximum utility attainable by the last facility is at most $\frac{n-(m-1)k}{2}+\frac{1}{2}$.
    Therefore the total utility of the optimal SW is $(m-1)k+\frac{n-(m-1)k}{2}+\frac{1}{2}$.
    If $n> (m+1)k$, the optimal SW on this instance is $mk$, and it is attained when $\floor{\frac{m}{2}}$ facilities are placed at $0$ and the others at $1$.
    To conclude the thesis it suffices to take the ratio of the optimal SW and the SW of the mechanism.
\end{proof}

\begin{remark}
    Notice that the highest approximation ratio occurs when $n=km+1$, in which case $ar(\PMp)=2-\frac{k}{mk+1}$.
Thus, as the number of facilities increases, we attain an approximation ratio that converges to $2$
\end{remark}.
%
%

Lastly, we observe that Side-By-Side mechanisms can be generalized to this case. 
It suffices to distribute all facilities between the positions of two consecutive agents.
As for the case where $m = 2$, the worst-case analysis of these mechanisms is conducted in a similar manner to the AIO mechanisms, yielding comparable results.

\section{The Higher Dimensional Problem}

In this section, we extend our study to higher dimensional settings.
We recall that an $m$-dimensional percentile mechanism is defined by a $d$ times $m$ matrix $V = (v^{(1)}, v^{(2)}, \dots, v^{(d)})$, where each row vector $v^{(j)}= (v^{(j)}_{1}, v^{(j)}_{2}, \dots, v^{(j)}_{m}) \in [0, 1]^m$ lies within the unit cube.
Given a the agents' reports $\vec x$, the percentile mechanism associated with $V$ determines the $k$-th facility's location by selecting, for each dimension $j \le d$, the $v^{(j)}_{k}$-th percentile of the ordered projection of $\vec x$ on the $j$-th dimension to be the coordinate of $k$-th facility in that dimension.
More formally, we have that
\[
    y^{(j)}=(z^{(1)}_{\floor{v^{(1)}_{j}(n-1)}+1},z^{(2)}_{\floor{v^{(2)}_{j}(n-1)}+1},\dots, z^{(d)}_{\floor{v^{(d)}_{j}(n-1)}+1}),
\]
where $z^{(l)}$ is the vector containing the projections of the agents positions on the $l$-coordinates re-ordered increasingly.
To summarize, a $d$-dimensional percentile mechanism determines the coordinates of the facilities by running $d$ percentile mechanism along each coordinate.
We call the percentile mechanism determining the $j$-th coordinates of the $m$ facilities the $j$-th coordinate percentile mechanism.
In what follows, we extend our framework to higher-dimensional problems and demonstrate that most of the previously presented results do not generalize to this setting.
This remains true even when we need to locate only two facilities among agents supported in a bidimensional space. 
For this reason and for the sake of simplicity, we limit our discussion to the case where two facilities need to be located, and each agent position is a point in the unit cube $[0,1]^2$.
%
%

\subsection{The Setting}

Let $\vec x\in ([0,1]^2)^n$ be the bi-dimensional vector containing the position of $n$ agents in the interval $[0,1]^2$, so that $\vec x=(x^{(1)},\dots,x^{(n)})$, where $x^{(i)}=(x_1^{(i)},x_i^{(i)})\in [0,1]^2$ is the vector containing the coordinates of the $i$-th agent.
We denote with $\vec k=(k_1,k_2)$ the $2$-dimensional vector containing all the capacities of the $2$ facilities, and assume that $k_1+k_2<n$.
In this case, a mechanism is a function $M:([0,1]^2)^n\to (\erre^2)^2$ that maps a vector containing the agents' reports to a facility location $\vec y=(y^{(1)},y^{(2)})$, where $y^{(j)}\in\erre^2$ is the position of the facility with capacity $k_j$.
As for the one dimensional case, once the facilities are located, agents compete in a First-Come-First-Served (FCFS) to determine their utilities.
The FCFS game is defined as in Section \ref{sec:settingstatement}, with the only difference being that the distance between an agent located at $\vec x^{(i)}$ and a facility located at $\vec y^{(j)}$ is 
\[
d(\vec x^{(i)},\vec y^{(j)})=\sqrt{\big(x^{(i)}_1-y^{(j)}_1\big)^2+\big(x^{(i)}_2-y^{(j)}_2\big)^2},
\]
so that the utility of every agent is either $0$ or $\sqrt{2}-d(\vec x^{(i)},\vec y^{(j)})$ depending on the outcome of the FCFS game.
It is then easy to see that the construction used to prove Theorem \ref{thm:NEexistance} trivially extends to this case, so that a pure Nash Equilibrium to the FCFS game always exists.
Lastly, given a Nash Equilibrium, namely $\gamma$, we define the Social Welfare associated with $\gamma$ as the sum of all the agent's utilities according to $\gamma$.

\subsection{The Equilibrium Stable Mechanisms}

We now discuss under which conditions a percentile mechanism is ES in the higher dimensional framework.
First of all, we notice that every ES percentile mechanism $M$ must be such that each Coordinate Percentile Mechanism (CPM) is ES.
Indeed, if the first CPM was not ES, it means that there exists a one dimensional instance $\vec x=(x_1,\dots,x_n)$ whose Nash Equilibria induce different Social Welfare values.
It is then easy to see that the percentile mechanism $M$ is not ES by considering the instance $\{x^{(i)}=(x_i,0)\}_{i\in [n]}$.
Unfortunately, assuming that each coordinate percentile mechanism is ES does not guarantee that the percentile mechanism is ES in higher dimensions.
In particular, enforcing Theorem \ref{thm:approximationratiom>2} to each coordinate percentile mechanism is not enough to characterize all the ES percentile mechanisms when $d>2$.
As the next example shows, percentile mechanisms whose coordinate percentile mechanisms are SBS or WG mechanisms no longer ensure that the mechanism is ES.

\begin{example}
\label{example2high}
    Let us consider a problem in which we have $10$ agents, that is $n=10$, and two facilities with capacity $2$, that is $k_1=k_2=2$.
    Let us consider the percentile mechanism induced by $v^{(1)}=v^{(2)}=(0.4,0.5)$, so that both the coordinate percentile mechanisms are SBS mechanisms.
    In what follows, we denote this mechanism with PM.
    Let us consider the following instance $x_1=x_2=x_3=(0,0)$, $x_4=(0.4,0.4)$, $x_5=(0.4,0.5)$, $x_6=(0.5,0.5)$, and $x_7=\dots=x_{10}=(0.8,0.8)$, so that PM places the facilities at $(0.4,0.4)$ and $(0.5,0.5)$.
    Notice that, depending on what strategy the agent at $(0.4,0.5)$ plays, the Social Welfare induced by the percentile mechanism changes.
\end{example}

\begin{figure}[t]
    \centering
    \begin{subfigure}[b]{0.45\textwidth}
        \centering
        \includegraphics[width=\textwidth]{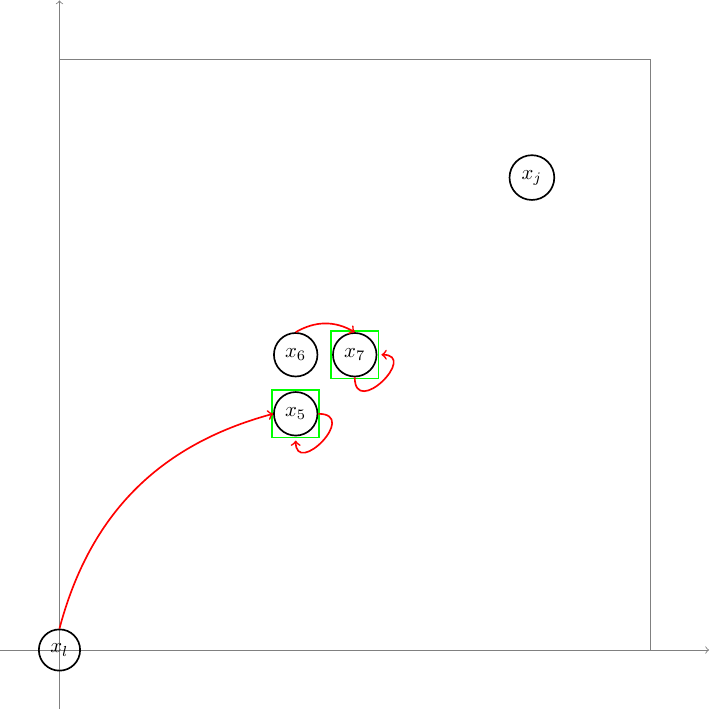}
    \end{subfigure}
    \hspace{0.02\textwidth} 
    \begin{subfigure}[b]{0.45\textwidth}
        \centering
        \includegraphics[width=\textwidth]{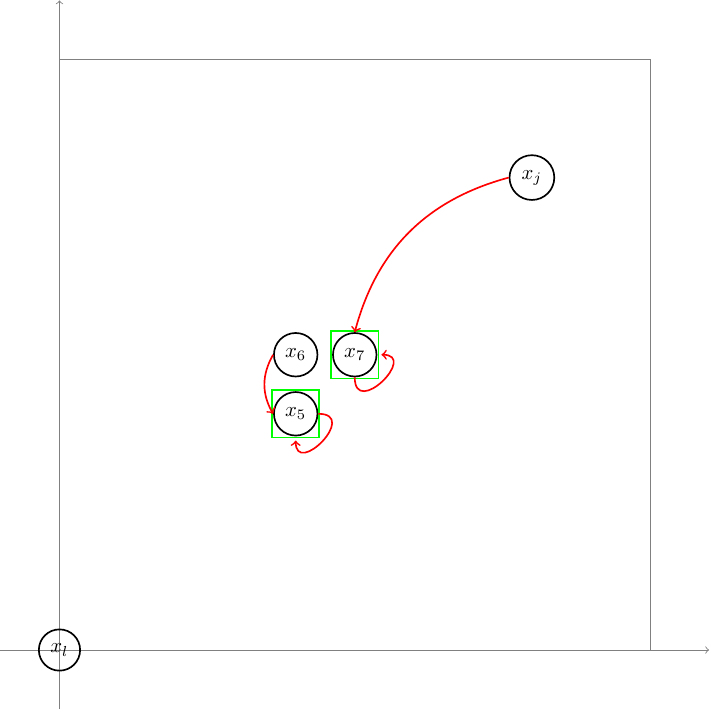}
    \end{subfigure}
    
    \caption{The two Nash Equilibria described in Example \ref{example2high}. The circles represents the agents (agents from $1$ to $4$ are represented by $x_l$ while agents from $8$ to $10$ are represented by $x_j$), the green squares the facilities and the red arrows the strategies played by the agents getting a non-null utility in the two different Nash Equilibria.}
    \label{fig:overall2}
\end{figure}

Example \ref{example2high}, can be extended to handle any percentile mechanism whose coordinates are either SBS or WG.
Indeed, the only percentile mechanisms that are ES for higher dimensional problem must determine one of the two coordinates using an AIO mechanism.
In particular, we now show that any percentile mechanism whose CPMs are two AIO mechanisms or AIO mechanism an SBS mechanism are ES.

\begin{theorem}
    A percentile mechanisms is ES if and only if all its CPMs are AIO mechanisms or all the CPMs but one are AIO mechanisms.
    In the latter case, the CPM that is not an AIO mechanism must be a SBS mechanism.
\end{theorem}

\begin{proof}
    If a percentile mechanism uses two AIO mechanisms to determine the coordinates of the facilities, we have that the mechanisms places all the facilities at the same spot, thus the mechanism is trivially ES.
    Let us consider the other case, in which one coordinate is determined via an AIO mechanism and the other one by an SBS mechanism.
    Without loss of generality, let us assume that the first coordinate is determined by the SBS, so that both the facilities have the same second coordinate, i.e. $y^{(1)}_2=y^{(2)}_2$.
    Let us define the following sets
    \begin{equation}
        A_1=\big\{x^{(i)}\;\text{s.t.}\; x^{(i)}_1\le y^{(1)}_1\big\}\quad\quad\text{and}\quad\quad A_2=\big\{x^{(i)}\;\text{s.t.}\; x^{(i)}_1\ge y^{(2)}_1\big\}.
    \end{equation}
    From our assumptions, we have that $A_1\cup A_2=\{x^{(i)}\}_{i\in [n]}$ and $A_1\cap A_2=\emptyset$.
    By definition of $A_1$ and $A_2$, we have that every agent in $A_1$ prefers $y^{(1)}$ to $y^{(2)}$.
    Vice-versa, every agent in $A_2$ prefers $y^{(2)}$ to $y^{(1)}$.
    We can then conclude the proof by using the same argument used to prove Theorem \ref{thm:SBS_ES}.
\end{proof}

Lastly, we show that a mechanism determining the facilities using an AIO and a WG mechanism is not ES.

\begin{example}
\label{ex_high2}
    Let us consider an instance for $n=10$, $k_1=k_2=2$, and the percentile mechanism induced by the vectors $v^{(1)}=(0,0.3)$ and $v^{(2)}=(0,0)$.
    Owing to Theorem \ref{thm:characterization2facilities}, we have that each coordinate percentile mechanism is ES, however the global percentile mechanism is not.
    Indeed, let us consider the following instance $x_1=(0,0)$, $x_2=(0.2,0)$, $x_3=(0.2,1)$, $x_4=(0.4,0)$, and $x_5=\dots=x_{10}=(0.7,0)$.
    The percentile mechanism places the facilities at $(0,0)$ and at $(0.4,0)$.
    This facility location induces two Nash Equilibria, one in which agent $x_3$ plays $1$ and one in which agent $x_3$ plays $2$.
    When agent $x_3$ plays $1$, we have that one of the agents located at $(0.5,0)$ is accommodated by $y_2$, thus the Social Welfare is $\sqrt{2}(3.5)$.
    If agent $x_3$ plays $2$, one of the agents located at $(0.7,0)$ is accommodate by the facility $y_1$, thus the Social Welfare is $\sqrt{2}(3.1)$.
    We then conclude that the percentile mechanism is not ES.
\end{example}

\begin{figure}[t]
    \centering
    \begin{subfigure}[b]{0.45\textwidth}
        \centering
        \includegraphics[width=\textwidth]{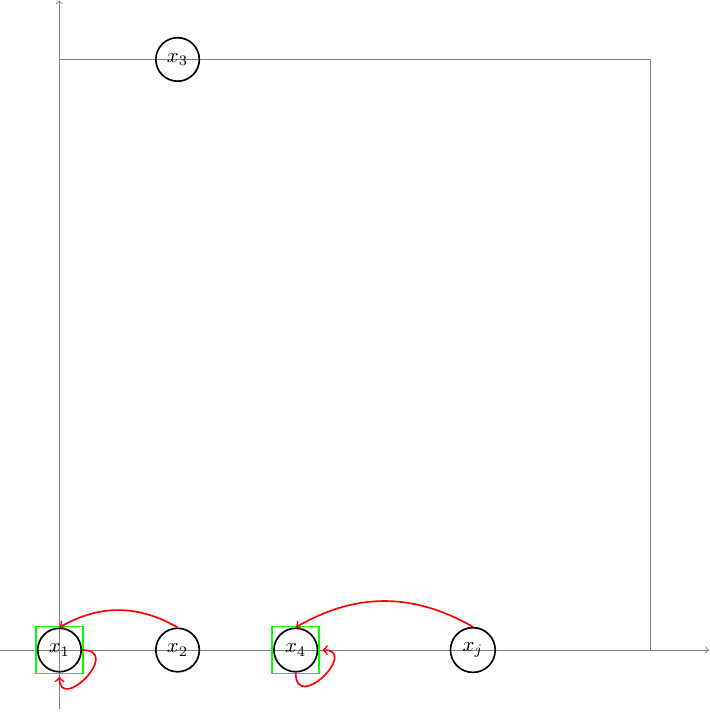}
    \end{subfigure}
    \hspace{0.02\textwidth} 
    \begin{subfigure}[b]{0.45\textwidth}
        \centering
        \includegraphics[width=\textwidth]{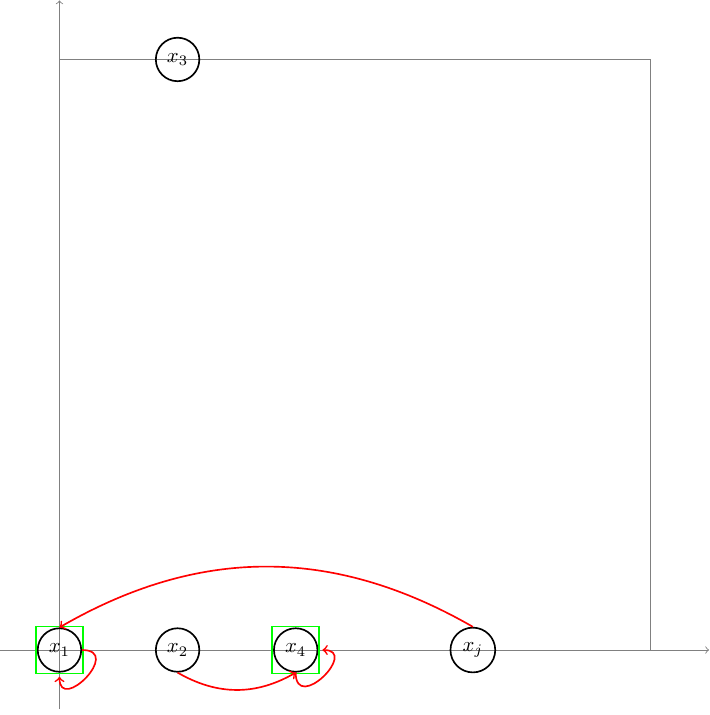}
    \end{subfigure}
    
    \caption{The two Nash Equilibria described in Example \ref{ex_high2}. The circles represents the agents (agents from $5$ to $10$ are represented by $x_j$), the green squares the facilities and the red arrows the strategies played by the agents getting a non-null utility in the two different Nash Equilibria.}
    \label{fig:overall3}
\end{figure}

Lastly, we characterize the approximation ratio of the ES percentile mechanism that determines the coordinates of both facilities using the median mechanism.
Indeed, by a similar argument used in Theorem \ref{thm:ar_med}, we have that the percentile mechanism induced by $V=((0.5,0.5),(0.5,0.5))$ achieves the lowest approximation ratio amongst all the possible ES higher dimensional percentile mechanisms.
%
%

\begin{theorem}
\label{thm:ar_highdim}
    Let $M$ be the percentile mechanism whose coordinate percentile mechanisms are median mechanisms, then we have that
    \begin{equation}
    \label{eq:thm14:1}
        ar(M)=\frac{\sqrt{2}}{\sqrt{2}-1}-\frac{\sqrt{2}}{\sqrt{2}-1}\frac{1}{(\sqrt{2}-1)(k_1+k_2)+1}\le \frac{\sqrt{2}}{\sqrt{2}-1},
    \end{equation}
    if $k_1,k_2\le\floor{\frac{n}{2}}$ and
    \begin{equation}
    \label{eq:thm14:2}
        ar(M)=\frac{\sqrt{2}(k_2+\floor{\frac{n}{2}})+\sqrt{2}-1}{(\sqrt{2}-1)(k_1+k_2)+1},
    \end{equation}
    otherwise.
\end{theorem}

\begin{proof}
    Let us consider an instance $\vec x=(x^{(1)},\dots, x^{(n)})\in ([0,1]^2)^n$. 
    For the sake of simplicity, we assume $n$ to be odd, so that the median along each coordinate is uniquely determined.
    Since all the coordinate percentile mechanisms of mechanism $M$ are AIO, the output of the mechanisms places both the facilities at the same location, which we denote with $y=(y_1,y_2)\in[0,2]^2$.
    Without loss of generality, let us assume that $0<y_1,y_2<1$ and that at most one agent is located at $y$.
    We then have that the set $[0,1]^2$ is divided into four regions $A_1=\{x\in[0,1]^2\;\;\text{s.t.}\; x_1\le y_1,\, x_2\le y_2\}$,  $A_2=\{x\in[0,1]^2\;\;\text{s.t.}\; x_1\le y_1,\, x_2> y_2\}$,  $A_3=\{x\in[0,1]^2\;\;\text{s.t.}\; x_1 > y_1,\, x_2\le y_2\}$, and  $A_4=\{x\in[0,1]^2\;\;\text{s.t.}\; x_1 > y_1,\, x_2 > y_2\}$ (see Figure \ref{fig:As}).
    We denote with $\mathcal{A}_i$ the set containing the agents positions in $A_i$, that is $\mathcal{A}_i=\{x^{(j)}\;\text{s.t.}\; x^{(j)}\in A_i\}$.
    By definition of $M$, we have that $|\mathcal{A}_1|+|\mathcal{A}_3|=|\mathcal{A}_2|+|\mathcal{A}_4|$ and $|\mathcal{A}_1|+|\mathcal{A}_2|=|\mathcal{A}_3|+|\mathcal{A}_4|$, which implies $|\mathcal{A}_2|=|\mathcal{A}_3|$ and $|\mathcal{A}_1|=|\mathcal{A}_4|$.

    \begin{figure}[t]
        \centering
        \includegraphics[width=0.6\textwidth]{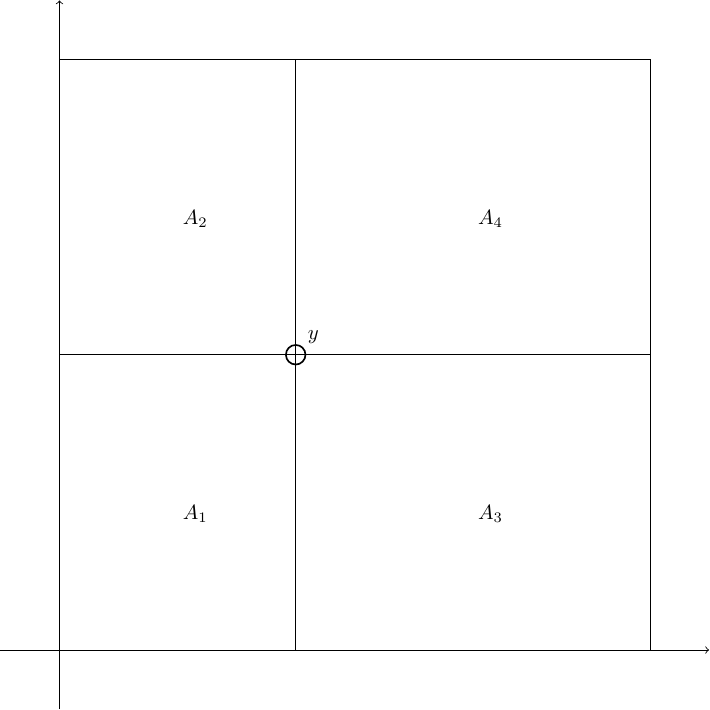}
        \caption{The way $A_1$, $A_2$, $A_3$, and $A_4$ partition of $[0,1]^2$ in Theorem \ref{thm:ar_highdim}.}
        \label{fig:As}
    \end{figure}

    Let us now assume that $|\mathcal{A}_2|=|\mathcal{A}_3|=0$, so that all the agents are located in $A_1$ and $A_4$.
    %
    %
    Since the number of agents is odd, we then have that the worst-case instance is such that $\floor{\frac{n}{2}}$ agents are located at $(0,0)$, $\floor{\frac{n}{2}}$ agents are located at $(1,1)$ and one agent is located at $y=(y_1,y_2)$.
    Notice that if $y_1+y_2\neq 1$, we can increase the approximation ratio by moving the facility (and the agent sharing the position with it) to a point of the line $x_1+x_2=1$.
    Finally, we have that the approximation ratio is maximized when either $y=(1,0)$ or $y=(0,1)$.
    Notice that if either $|\mathcal{A}_2|,|\mathcal{A}_3|\neq 0$, we can find an instance such that the Social Welfare of the mechanism is the same while the optimal Social Welfare is equal or larger than the SW of the original instance.
    Therefore, the worst-case instance for mechanism $M$ is such that $\floor{\frac{n}{2}}$ agents are located at $(0,0)$, $\floor{\frac{n}{2}}$ agents are located at $(1,1)$, and one agent is located at $(0,1)$.
    Notice that the SW induced by the mechanism on this instance is $(\sqrt{2}-1)(k_1+k_2)+1$.
    To conclude the proof we compute the approximation ratio of the mechanism, to do that we consider two cases.

    \textbf{Case $1$: $k_2,k_1\le \floor{\frac{n}{2}}$.} In this case, the optimal SW is equal to $k_1+k_2$, thus we infer \eqref{eq:thm14:1}.

    \textbf{Case $2$: $k_1 > \floor{\frac{n}{2}}$.} In this case, the optimal SW is equal to $k_1+\floor{\frac{n}{2}}+(\sqrt{2}-1)$, thus we infer \eqref{eq:thm14:2} and conclude the thesis.
\end{proof}

\subsection{Beyond \texorpdfstring{$m=2$}{Lg} and \texorpdfstring{$d=2$}{Lg}}

To conclude, we briefly comment on the cases in which either we have more than two facilities to locate, \textit{i.e.} $m>2$, or the agents inhabit a space whose dimension is higher than $2$, \textit{i.e.} $d>2$.
For higher dimensional problems, it is possible to adapt Example \ref{example2high} and Example \ref{ex_high2} to show that the only percentile mechanisms that are ES must be such that $d-1$ coordinates of the mechanisms are AIO mechanisms and the last one is either an AIO mechanism or an SBS mechanism.
When $m>2$, we have that the results presented in this section still hold: the only percentile mechanisms that are ES either place all the facilities at the same spot or they place the two facilities at two positions whose coordinate differ only in one entry.
In the first case, we have that the coordinates of the percentile mechanism are all AIO mechanisms, in the second case, the coordinates of the mechanism are all AIO except for one coordinate which is an SBS.

\section{Experimental Results}

In this section, we complement our theoretical study of the $m$-CFLP with scarce resources by running several numerical experiments.
In particular, we want to evaluate the performances of the percentile mechanisms identified by Theorem \ref{thm:bestPMPmechanism} when the inputs of the mechanisms are generated by a probability distribution.
To this extent, we consider two different metrics to assess the efficiency of a mechanism.
The first metric we consider is the Bayesian approximation ratio, which measures how close the expected SW induced by the mechanism and the expected optimal SW are when the agents' positions are samples drawn from a probability distribution $\mu$, \cite{hartline2013bayesian}.
More formally, the Bayesian approximation of $\PMp$ is 
\begin{equation}
    B_{ar}(\PMp):=\frac{\EE_{\vec X\sim \mu}[SW_{opt}(\vec X)]}{\EE_{\vec X\sim \mu}[SW_{\PMp}(\vec X)]},
\end{equation}
where $\vec X$ is a $n$ dimensional random vector distributed according to $\mu$.
The second metric we consider is the Average-Case approximation ratio, which measures the average ratio between the optimal Social Welfare and the Social Welfare induced by the mechanism when the agents' positions are samples drawn from a probability distribution $\mu$, \cite{DBLP:journals/iandc/DengGZ22}.
More formally, the Average-Case approximation ratio is defined as
\begin{equation}
    AVG_{ar}(\PMp)=\EE_{X\sim\mu}\Big[\frac{SW_{opt}(\vec x)}{SW_{\PMp}(\vec x)}\Big].
\end{equation}

It is worthy of notice that the Bayesian approximation ratio $B_{ar}$ and the Average-Case approximation ratio $AVG_{ar}$ are two different measures: $B_{ar}$ measures the percentage-loss of the Social Welfare, while the $AVG_{ar}$ measures the average Social Welfare percentage-loss.
Our aim is to show that percentile mechanisms that are optimal according to the worst-case analysis, namely $\mathcal{PM}_{best}$ (see Theorem \ref{thm:bestPMPmechanism}), are optimal or quasi-optimal with respect to both the Bayesian approximation ratio and the Average-Case ratio.
For this reason, we run two tests:
\begin{itemize}
    \item first, we assess to what extent the Bayesian approximation ratio and the Average-Case ratio depend on the percentile vector inducing the percentile mechanism when all the agents are independent and identically distributed (i.i.d.).
    In particular, we compare $\mathcal{PM}_{best}$ with $\mathcal{PM}_{(0,1)}$, that is the percentile mechanism induced by $\vec v=(0,1)$.
    \item Secondly, we assess the Bayesian approximation ratio and Average-Case ratio of $\mathcal{PM}_{best}$ when diverse agents within the populations follow distinct distributions. 
    This examination helps determine the suitability of $\mathcal{PM}_{best}$ for addressing problems involving non-identically distributed agents. 
\end{itemize}

We run our experiments for different distributions $\mu$ and different capacity vectors $\vec k$ in order to provide a comprehensive view.
Moreover, since the highest approximation ratio is attained when $m=2$, we only consider cases in which we need to place two facilities.
Throughout our experiments, we sample the agents' positions from three different probability distributions supported over $[0,1]$: 
\begin{enumerate*}[label=(\roman*)]
    \item the uniform distribution, namely $\mathcal{U}$ whose density is equal to $1$ over a $[0,1]$,

    \item the triangular distributions, namely $\mathcal{T}$, whose density is equal to $2(1-x)$ over a $[0,1]$, and
    
    \item the Beta distributions of parameters $\alpha,\beta>0$, namely $\mathcal{B}(\alpha,\beta)$ whose density is equal to $Cx^{\alpha-1}(1-x)^{\beta-1}$ over a $[0,1]$, where $C$ is a normalizing constant.
\end{enumerate*} 
We consider different capacity vectors $\vec k$. 
Specifically, we consider balanced capacities $\vec{k}=(k,k)$ and unbalanced capacities $\vec{k}=(k_1, k_2), k_1 > k_2$. For the case of balanced capacities, we consider $k=\alpha n$, where $\alpha=0.1,0.2,0.3$, and $0.4$. 
For the case of unbalanced capacities, we consider the slightly unbalanced capacities i.e. $\vec{k}=(0.4n,0.3n)$, and highly unbalanced capacities i.e. $\vec{k}=(0.6n,0.2n)$, $(0.7n,0.1n)$.
Lastly, for every instance $\vec x$, we do not compute the optimal SW, but rather an upper bound to that quantity.
Indeed, the optimal position of the facilities can be any couple of points in $[0,1]$.
Furthermore, to select the optimal facility location we must compute all the NE of every possible facility location and select the one achieving the highest SW.
For these reasons, we consider an easier to compute upper bound that is obtained by considering the maximum SW achievable when the mechanism forces agents to use a specific facility, that is
\begin{equation}
    \label{eq:futurework}
    SW_{UB}(\vec x):=\sup_{y_1,y_2\in[0,1]}\sup_{\pi\in\Pi}\sum_{i\in[n_u]}\sum_{j\in[2]}(1-|x_i-y_j|)\pi_{i,j}
\end{equation}
where $\Pi$ is the set containing all the matching $\pi\subset [n]\times [2]$ such that \begin{enumerate*}[label=(\roman*)]
    \item every $i\in [n]$ has degree that is equal or lower than $1$  and 
    \item every $j\in[2]$ has degree equal to $k_j$.
\end{enumerate*}
This quantity is easy to compute, as the set of optimal positions for the facilities coincides with the positions of the agents.

\subsection{Experiment results -- Comparing different percentile mechanisms.} 
In this experiment we want to assess to what degree the percentile vector affects the performances of the percentile mechanisms it induces.
For this reason, we compare the empiric Bayesian approximation ratio and the Average-Case approximation ratio of the best percentile mechanism $\mathcal{PM}_{best}$ (identified via Theorem \ref{thm:bestPMPmechanism}), with the respective ratios of a mechanism that places the facilities at the extreme agents' positions. That is, the percentile mechanism induced by the vector $\vec w=(0,1)$, namely $\mathcal{PM}_{(0,1)}$.
We first consider the case of balanced capacities $\vec k=(k,k)$.
Figure~\ref{fig:battleofvectors} shows the average and the $95\%$ confidence interval (CI) of Bayesian approximation ratio for $n=10,20,30,40,50$ when the agents are distributed according to $\mathcal{U}$ and $\mathcal{T}$.
In Figure \ref{fig:battleofvectors_avg} we plot the results for the Average-Case Ratio under the same settings.
Each average is computed over $500$ instances.
We observe that the percentile mechanism identified in Theorem \ref{thm:bestPMPmechanism} achieves the a Bayesian approximation ratio that is lower than the one obtained by $\mathcal{PM}_{(0,1)}$ for every value of $n$. 
Moreover, the Bayesian approximation ratio of $\mathcal{PM}_{(0,1)}$ consistently increases as the number of agents increases, while the Bayesian approximation ratio of $\mathcal{PM}_{best}$ remains constant regardless of $n$.
Remarkably, these comments hold true even for th Average-Case approximation ratio, showing that the best percentile mechanism is quasi-optimal for both the metrics under study.

\begin{figure}[t]
  \centering

    \begin{subfigure}{0.32\linewidth}
    \includegraphics[width=\linewidth]{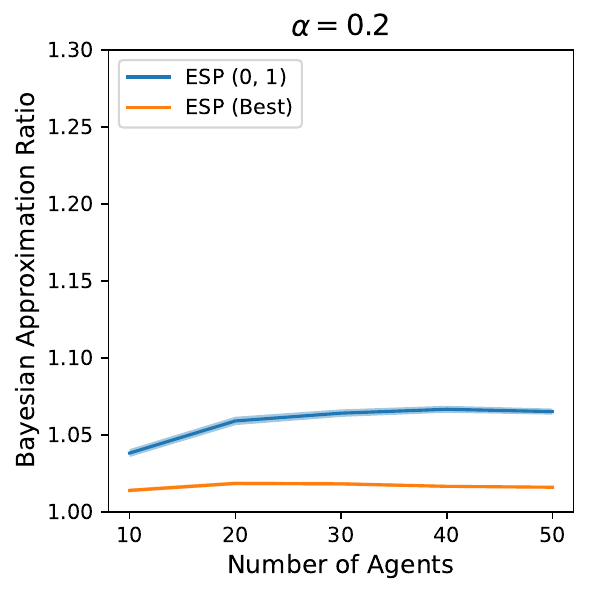}
  \end{subfigure}
  \begin{subfigure}{0.32\linewidth}
    \includegraphics[width=\linewidth]{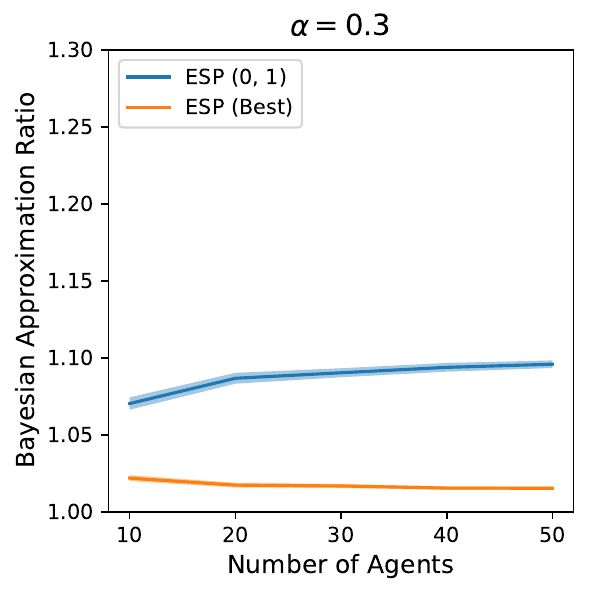}
  \end{subfigure}
  \begin{subfigure}{0.32\linewidth}
    \includegraphics[width=\linewidth]{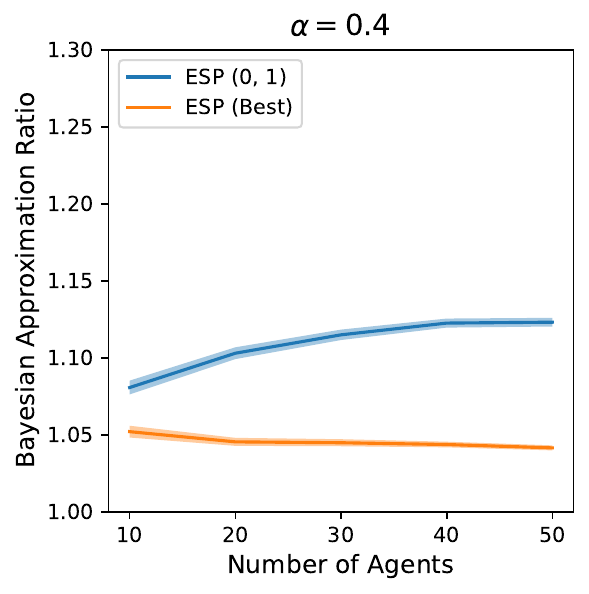}
  \end{subfigure}

    \medskip 

  \begin{subfigure}{0.32\linewidth}
    \includegraphics[width=\linewidth]{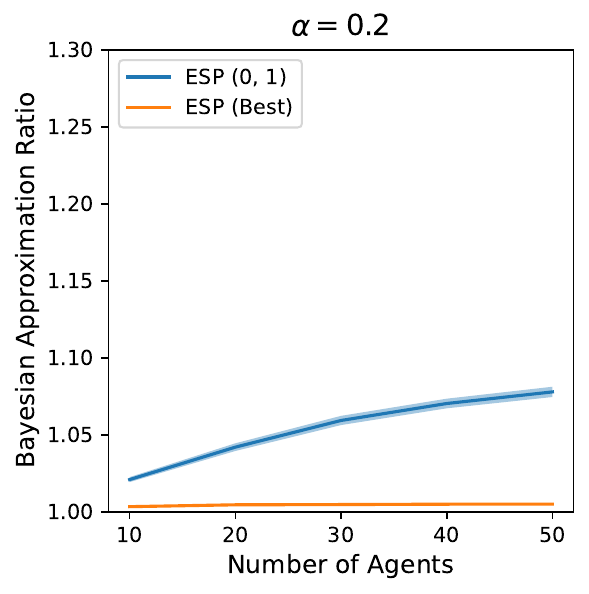}
  \end{subfigure}
  \begin{subfigure}{0.32\linewidth}
    \includegraphics[width=\linewidth]{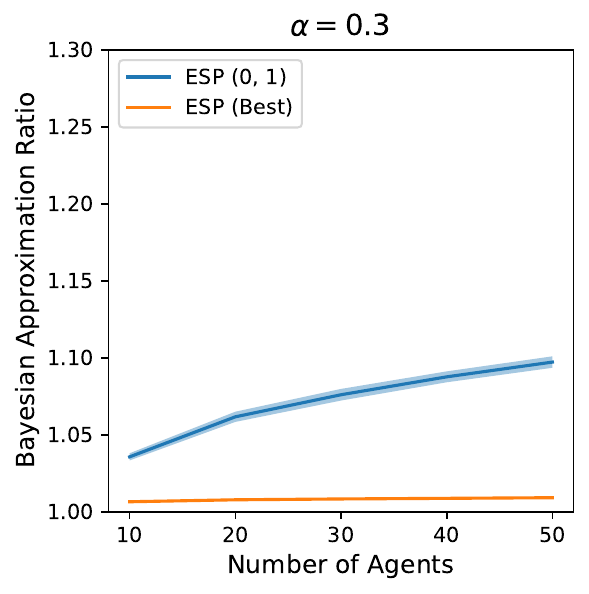}
  \end{subfigure}
  \begin{subfigure}{0.32\linewidth}
    \includegraphics[width=\linewidth]{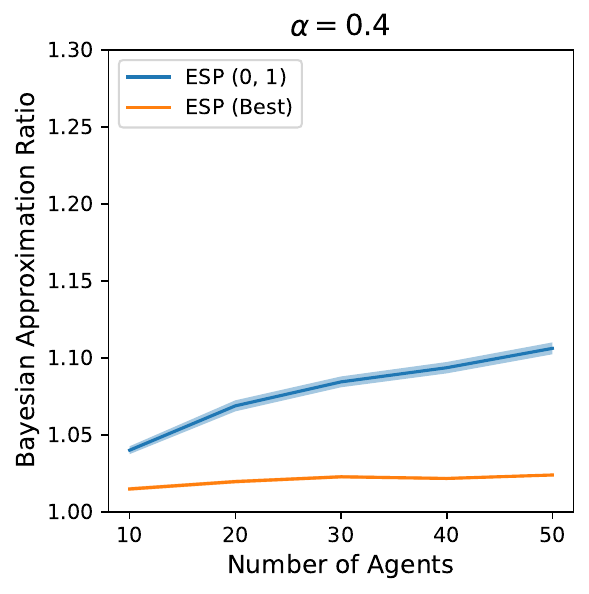}
  \end{subfigure}

  \medskip 

  \begin{subfigure}{0.32\linewidth}
    \includegraphics[width=\linewidth]{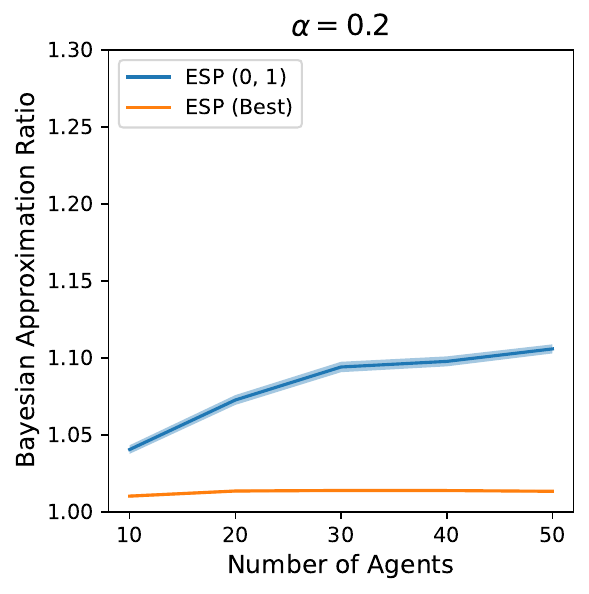}
  \end{subfigure}
  \begin{subfigure}{0.32\linewidth}
    \includegraphics[width=\linewidth]{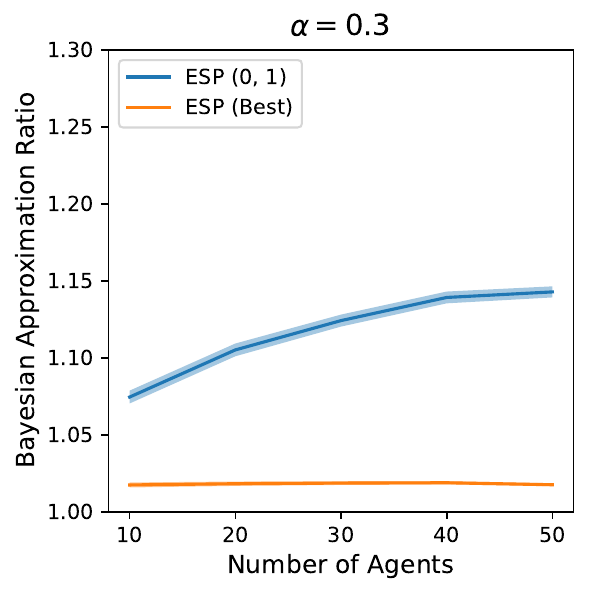}
  \end{subfigure}
  \begin{subfigure}{0.32\linewidth}
    \includegraphics[width=\linewidth]{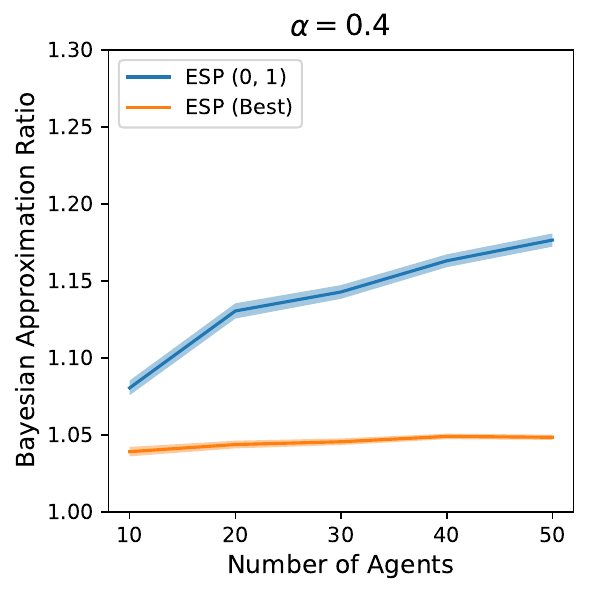}
  \end{subfigure}

  \caption{The Bayesian approximation ratio of $\mathcal{PM}_{best}$ and $\mathcal{PM}_{\vec w}$ in the balanced case, i.e. $k_1=k_2$ for $n=10,20,\dots,50$.
    Every column contains the results for different vector $\vec k$.
    %
    %
    In the first row, we consider a uniform distribution.
    In the second row a symmetric Beta, that is $\mathcal{B}(5,5)$.
    The third row contains the results for the triangular distribution $\mathcal{T}$. 
    %
    \label{fig:battleofvectors}}
\end{figure}

We now consider the case in which the facilities have difference capacities.
Figure~\ref{fig:battleofvectors_asym} shows the average and the $95\%$ confidence interval (CI) of Bayesian approximation ratio for $n=10,20,30,40,50$ when the capacities of the two facilities are not balanced, and the agents are distributed according to $\mathcal{T}$.
In Figure \ref{fig:battleofvectors_asym_avg} we plot the results for the Average-Case Ratio under the same settings.
More specifically, we consider $\vec k=(0.4n,0.3n)$ and $(0.7n,0.1n)$.
We observe that the percentile mechanism identified by Theorem \ref{thm:bestPMPmechanism} has a lower and more stable Bayesian approximation ratio, highlighting how choosing a percentile vector affects the average performances of the mechanism. 
The same holds for the Average-Case approxiamtion ratio.

\begin{figure}[t]
  \centering

    \begin{subfigure}{0.32\linewidth}
    \includegraphics[width=\linewidth]{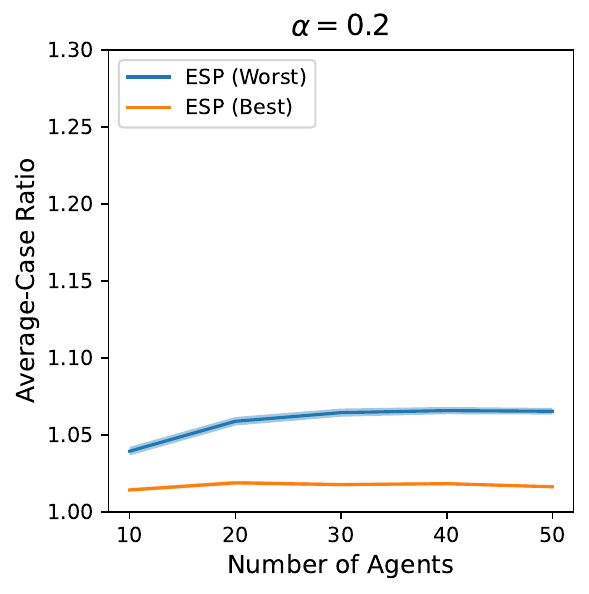}
  \end{subfigure}
  \begin{subfigure}{0.32\linewidth}
    \includegraphics[width=\linewidth]{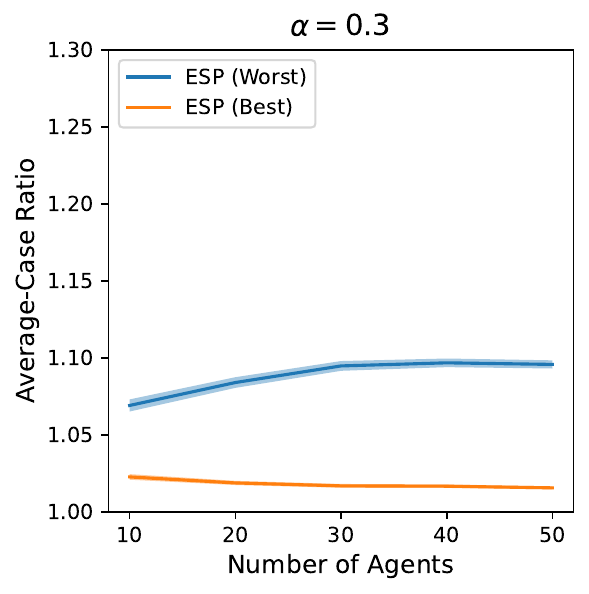}
  \end{subfigure}
  \begin{subfigure}{0.32\linewidth}
    \includegraphics[width=\linewidth]{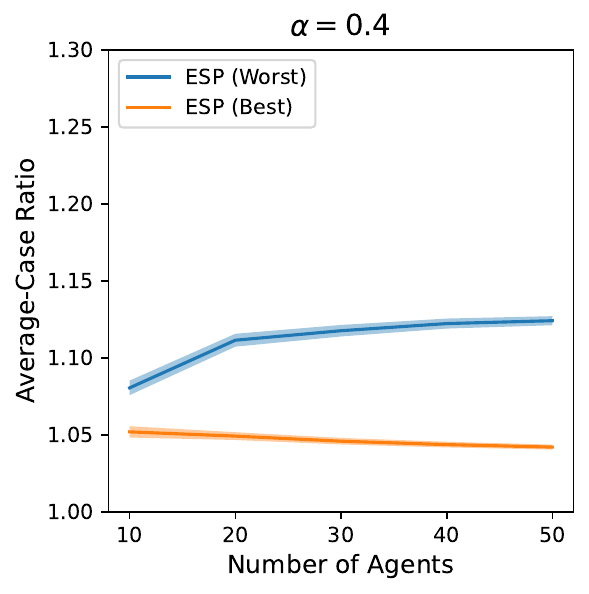}
  \end{subfigure}

    \medskip 

  \begin{subfigure}{0.32\linewidth}
    \includegraphics[width=\linewidth]{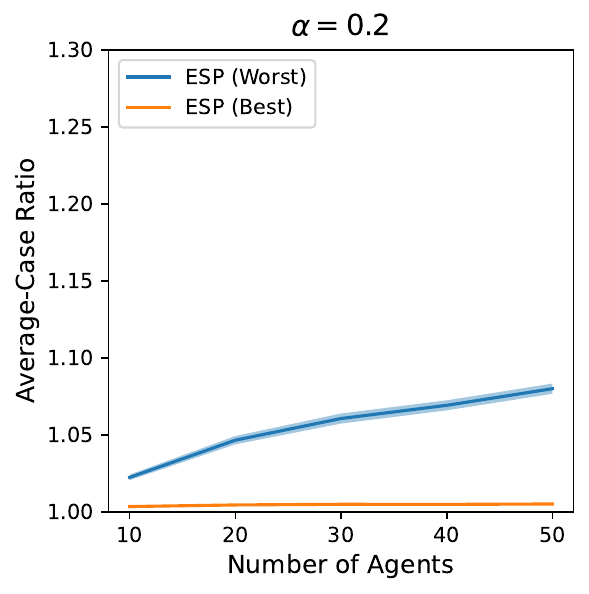}
  \end{subfigure}
  \begin{subfigure}{0.32\linewidth}
    \includegraphics[width=\linewidth]{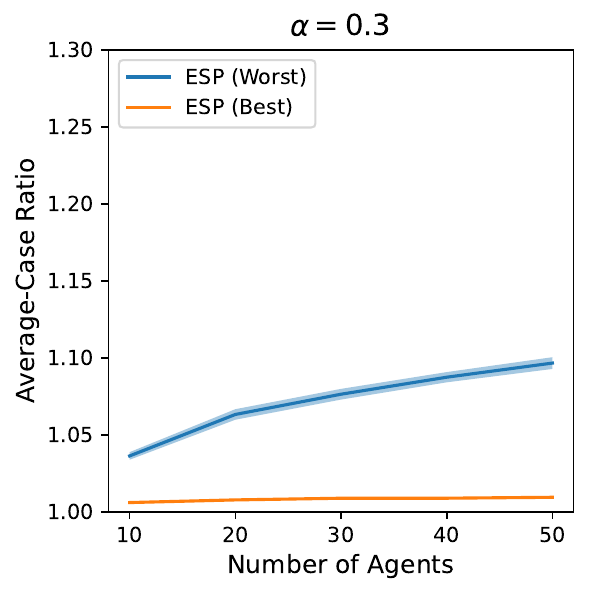}
  \end{subfigure}
  \begin{subfigure}{0.32\linewidth}
    \includegraphics[width=\linewidth]{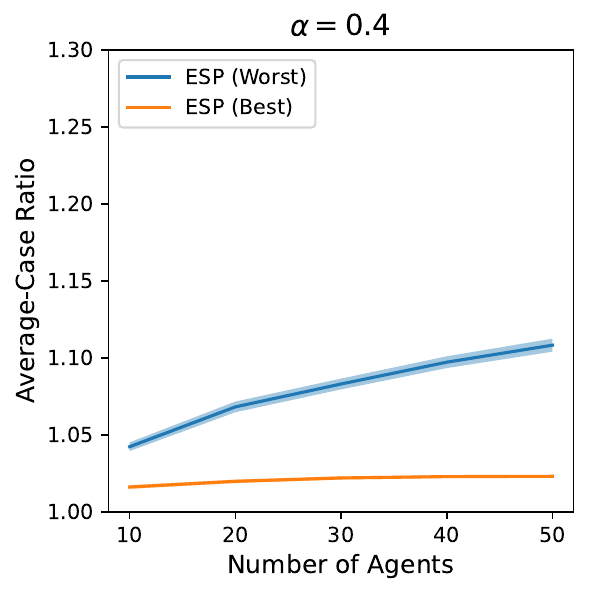}
  \end{subfigure}

  \medskip 

  \begin{subfigure}{0.32\linewidth}
    \includegraphics[width=\linewidth]{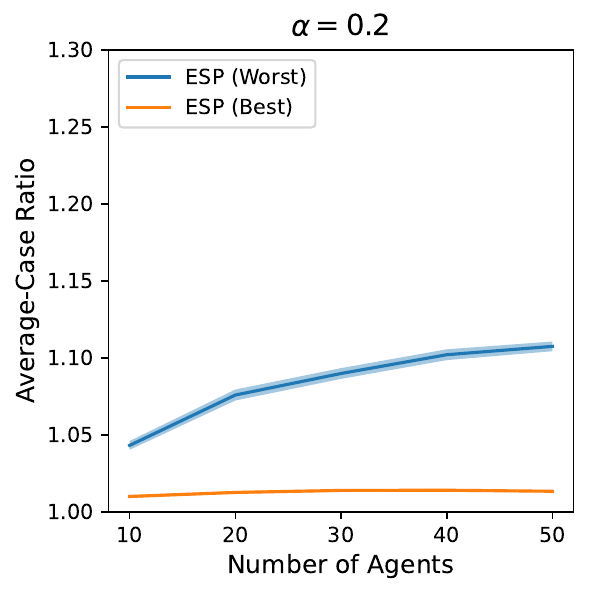}
  \end{subfigure}
  \begin{subfigure}{0.32\linewidth}
    \includegraphics[width=\linewidth]{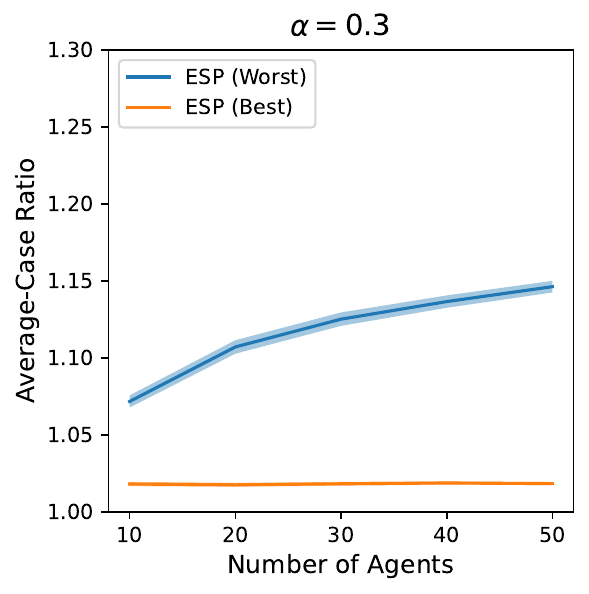}
  \end{subfigure}
  \begin{subfigure}{0.32\linewidth}
    \includegraphics[width=\linewidth]{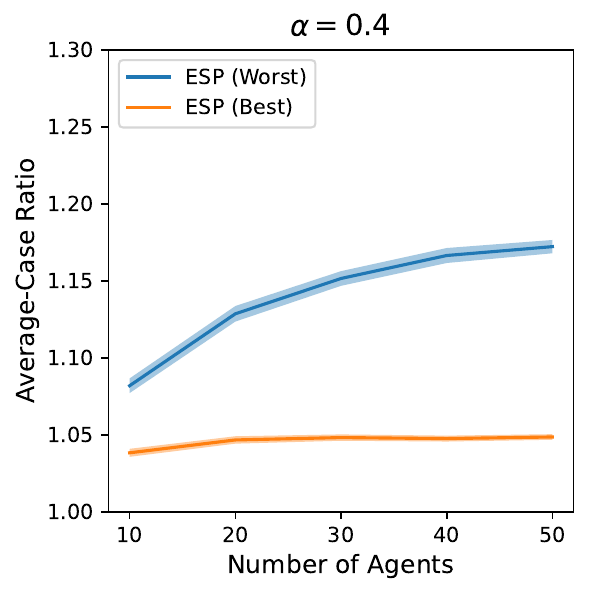}
  \end{subfigure}

  \caption{The Average-Case approximation ratio of $\mathcal{PM}_{best}$ and $\mathcal{PM}_{\vec w}$ in the balanced case, i.e. $k_1=k_2$ for $n=10,20,\dots,50$.
    Every column contains the results for different vector $\vec k$.
    In the first row, we consider a uniform distribution.
    In the second row a symmetric Beta, that is $\mathcal{B}(5,5)$.
    The third row contains the results for the triangular distribution $\mathcal{T}$. 
    \label{fig:battleofvectors_avg}}
\end{figure}



  

\begin{figure}
  \centering

  \begin{subfigure}{0.31\linewidth}
    \includegraphics[width=\linewidth]{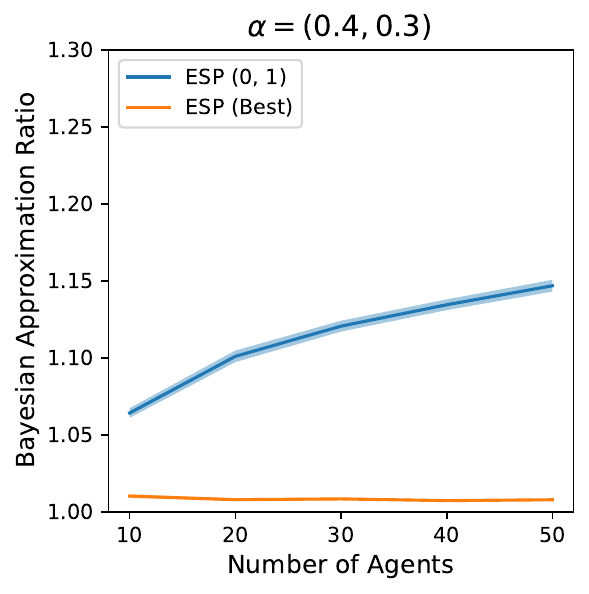}
  \end{subfigure}
  \begin{subfigure}{0.31\linewidth}
    \includegraphics[width=\linewidth]{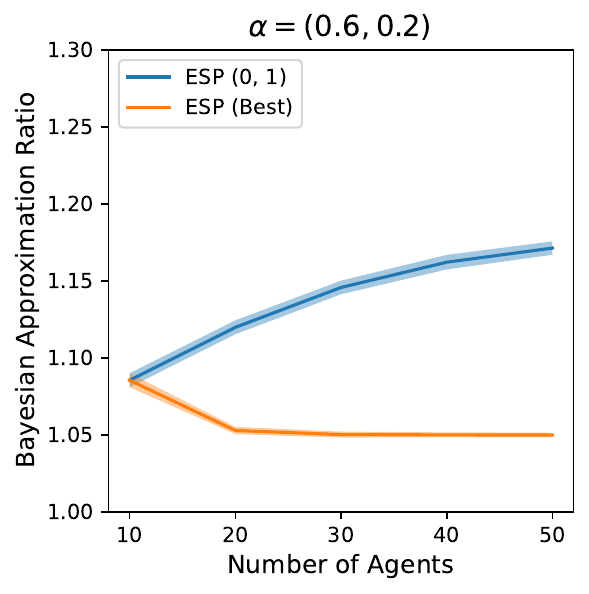}
  \end{subfigure}
  \begin{subfigure}{0.31\linewidth}
    \includegraphics[width=\linewidth]{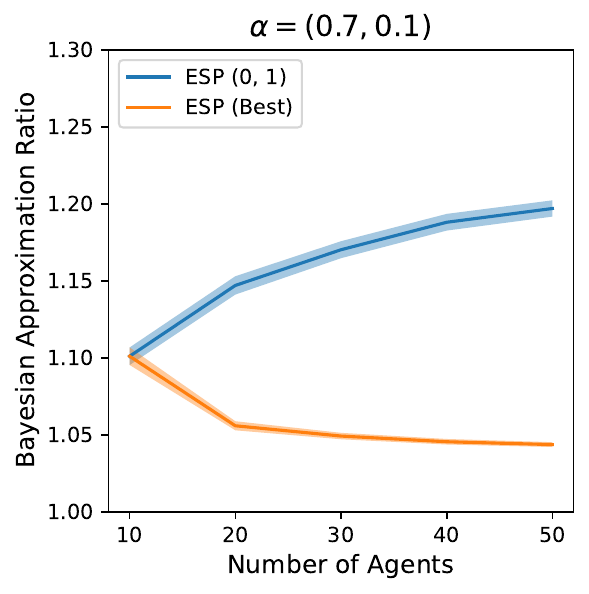}
  \end{subfigure}

  \medskip 

    \begin{subfigure}{0.31\linewidth}
    \includegraphics[width=\linewidth]{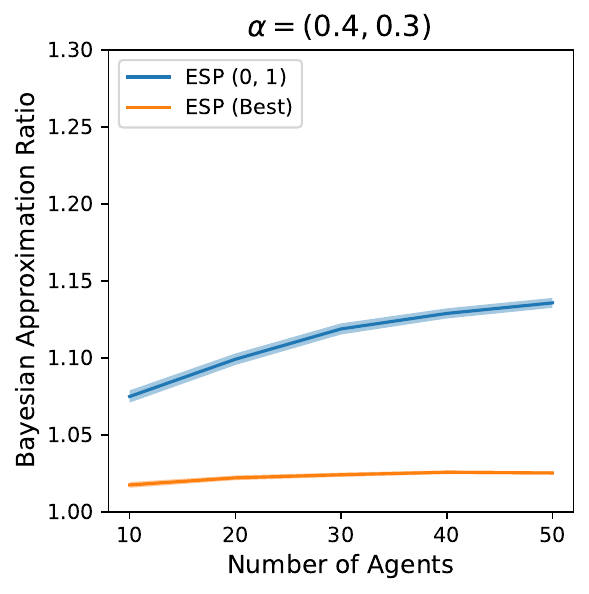}
  \end{subfigure}
  \begin{subfigure}{0.31\linewidth}
    \includegraphics[width=\linewidth]{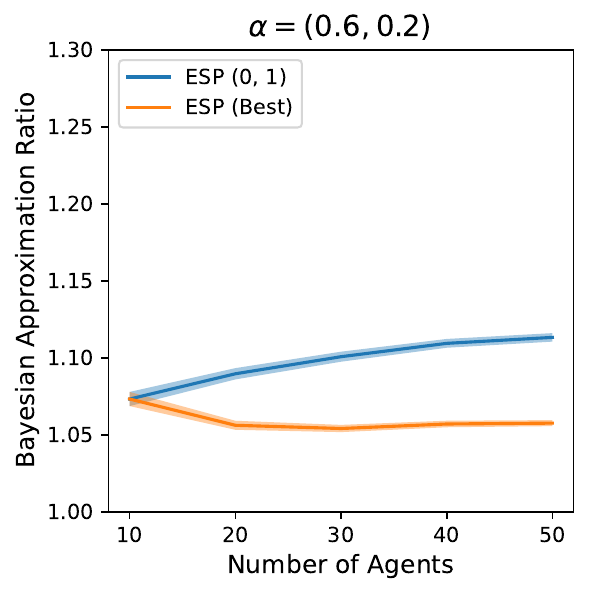}
  \end{subfigure}
  \begin{subfigure}{0.31\linewidth}
    \includegraphics[width=\linewidth]{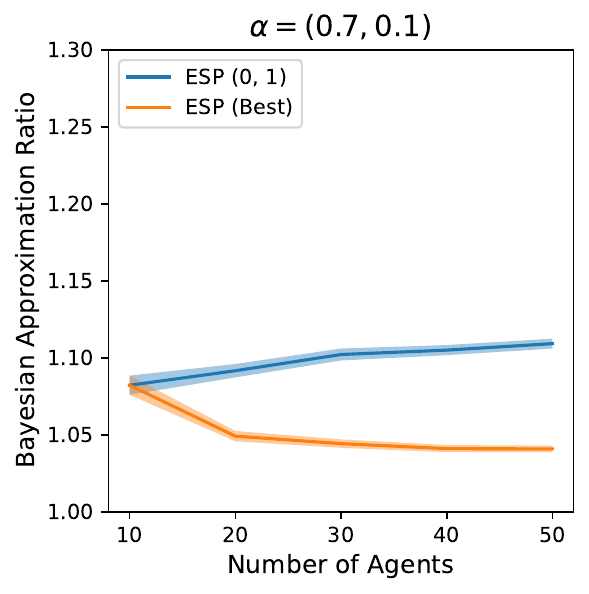}
  \end{subfigure}

  \medskip


  \begin{subfigure}{0.31\linewidth}
    \includegraphics[width=\linewidth]{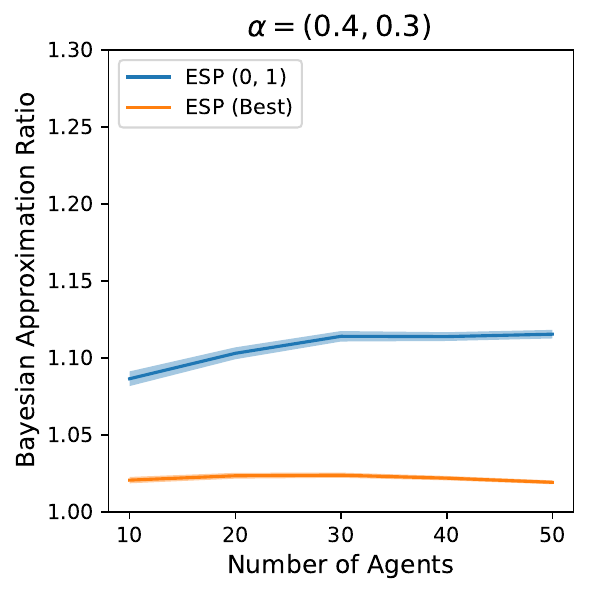}
  \end{subfigure}
  \begin{subfigure}{0.31\linewidth}
    \includegraphics[width=\linewidth]{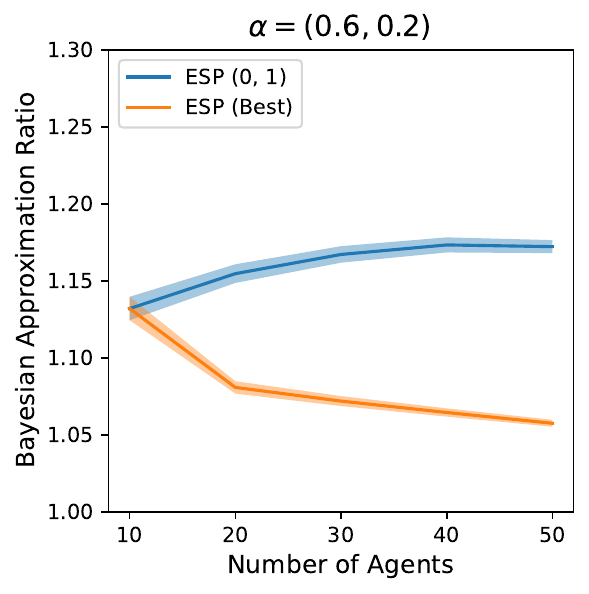}
  \end{subfigure}
  \begin{subfigure}{0.31\linewidth}
    \includegraphics[width=\linewidth]{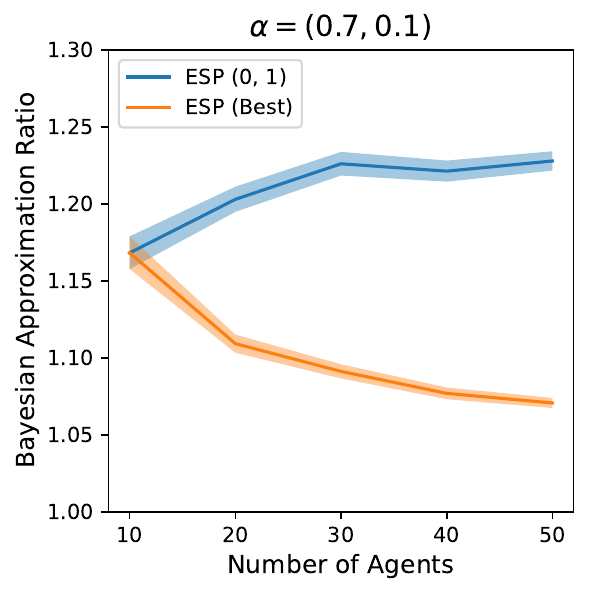}
  \end{subfigure}
  
\caption{The Bayesian approximation ratio of $\mathcal{PM}_{best}$ and $\mathcal{PM}_{\vec w}$ when the agents are distributed according to $\mathcal{T}$ and the facilities are unbalanced, i.e. $k_1=\alpha_1 n\neq k_2=\alpha_2 n$ for $n=10,20,\dots,50$.
Every column contains the results for different vector $\vec k$.
The first row contains the results for a symmetric Beta distribution, that is $\mathcal{B}(5,5)$.
The second row contains the results for the triangular distribution $\mathcal{T}$. 
The last row contains the results for the Uniform distribution $\mathcal{U}$.\label{fig:battleofvectors_asym}}
\end{figure}

\begin{figure}
  \centering

  \begin{subfigure}{0.31\linewidth}
    \includegraphics[width=\linewidth]{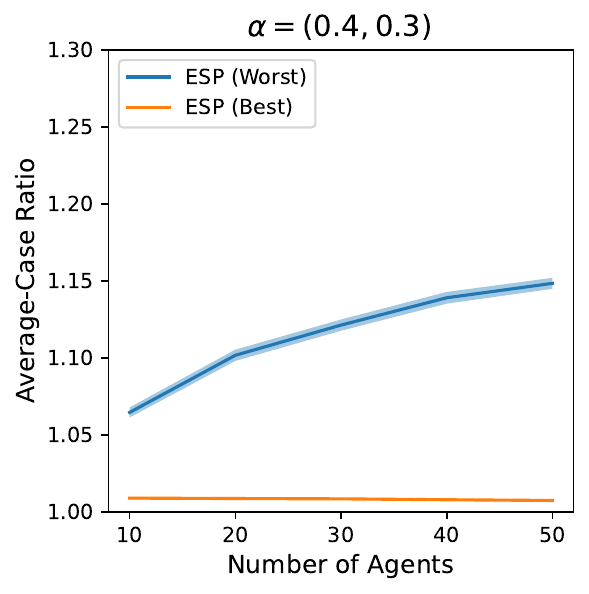}
  \end{subfigure}
  \begin{subfigure}{0.31\linewidth}
    \includegraphics[width=\linewidth]{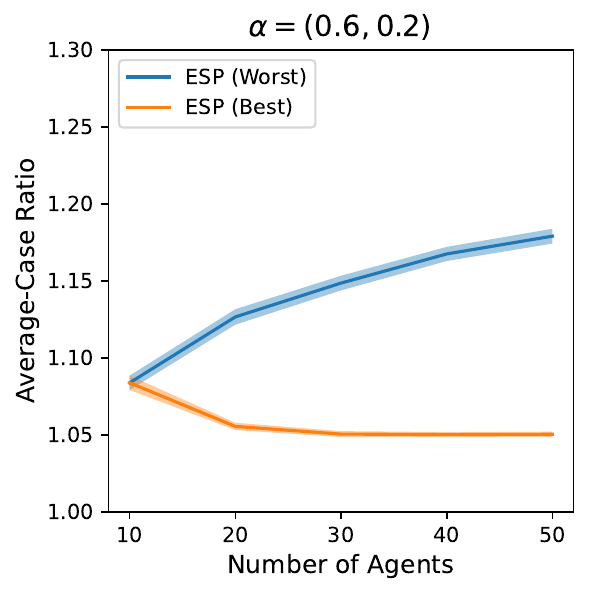}
  \end{subfigure}
  \begin{subfigure}{0.31\linewidth}
    \includegraphics[width=\linewidth]{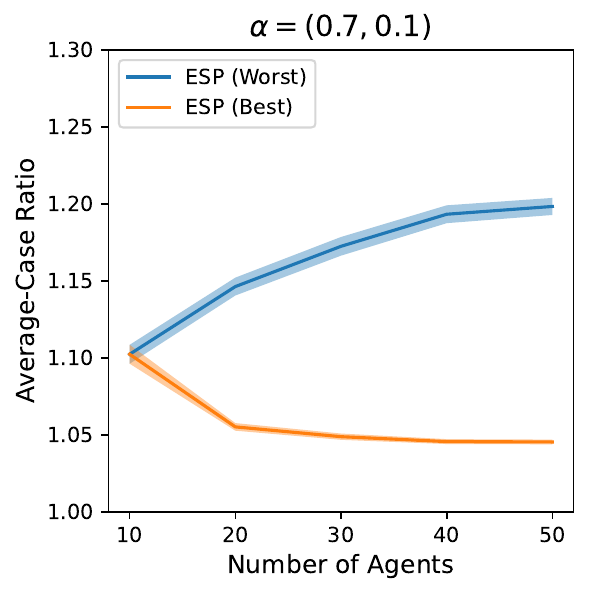}
  \end{subfigure}

  \medskip 

    \begin{subfigure}{0.31\linewidth}
    \includegraphics[width=\linewidth]{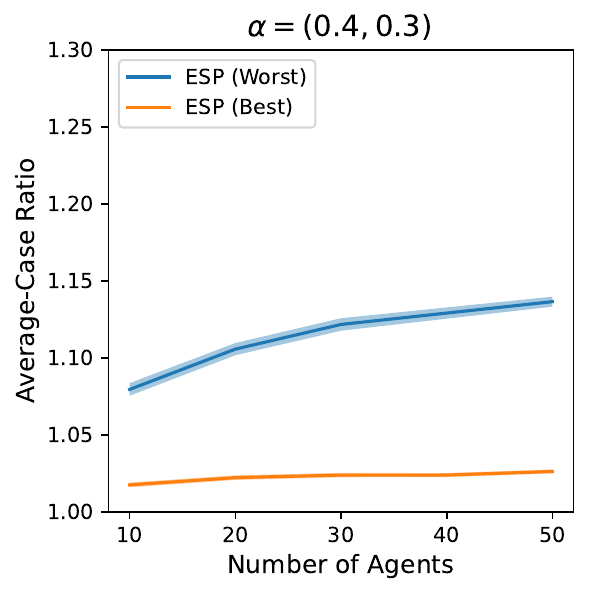}
  \end{subfigure}
  \begin{subfigure}{0.31\linewidth}
    \includegraphics[width=\linewidth]{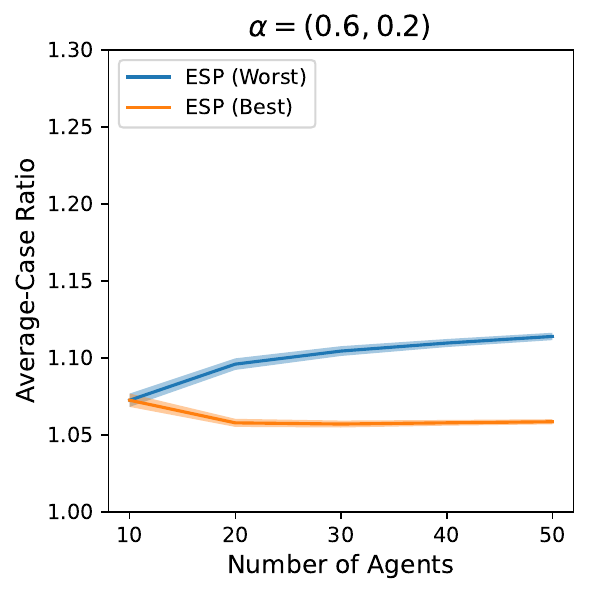}
  \end{subfigure}
  \begin{subfigure}{0.31\linewidth}
    \includegraphics[width=\linewidth]{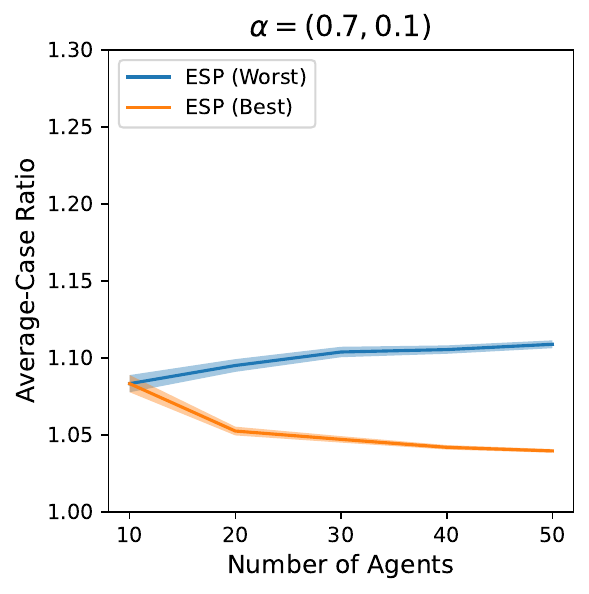}
  \end{subfigure}

  \medskip


  \begin{subfigure}{0.31\linewidth}
    \includegraphics[width=\linewidth]{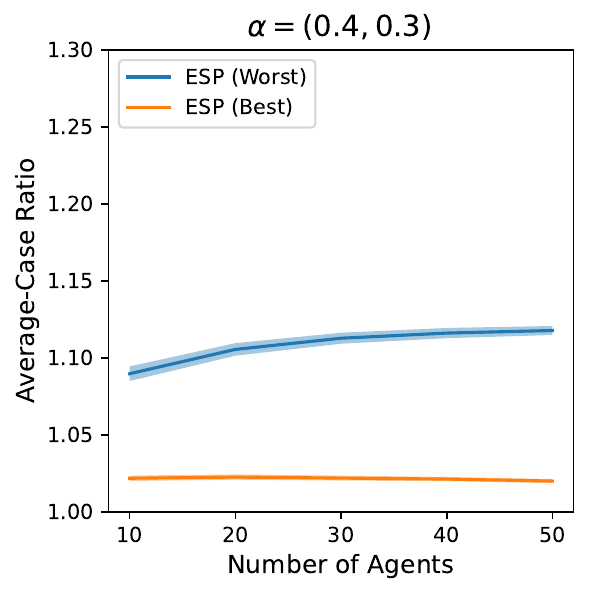}
  \end{subfigure}
  \begin{subfigure}{0.31\linewidth}
    \includegraphics[width=\linewidth]{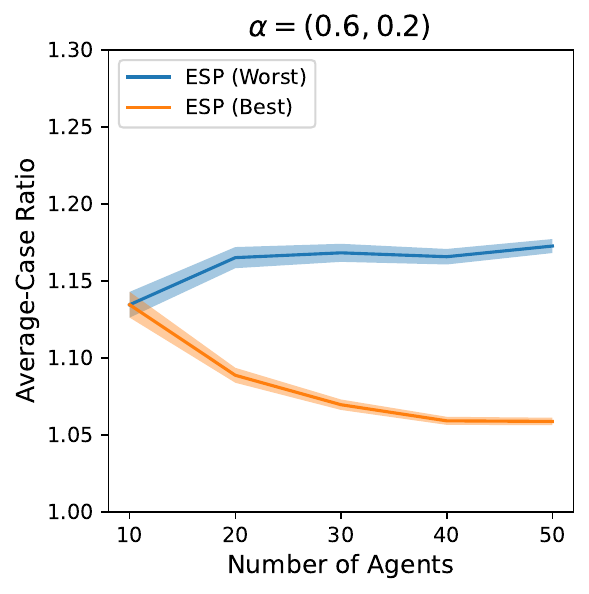}
  \end{subfigure}
  \begin{subfigure}{0.31\linewidth}
    \includegraphics[width=\linewidth]{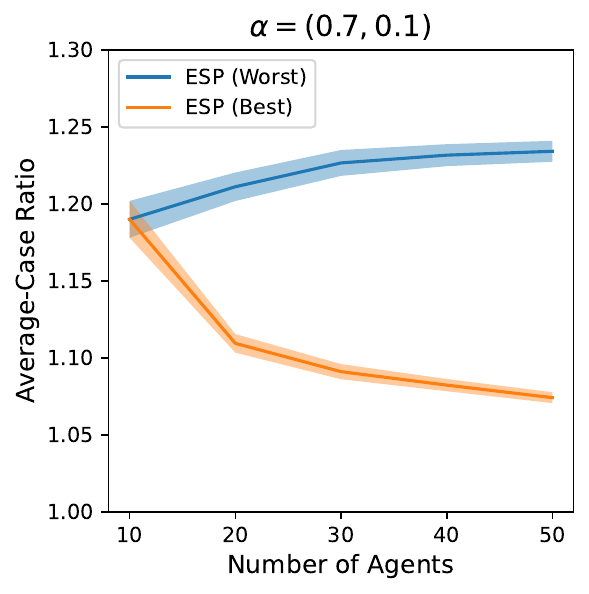}
  \end{subfigure}
  
  \caption{The Average-Case approximation ratio of $\mathcal{PM}_{best}$ and $\mathcal{PM}_{\vec w}$ when the agents are distributed according to $\mathcal{T}$ and the facilities are unbalanced, i.e. $k_1=\alpha_1 n\neq k_2=\alpha_2 n$ for $n=10,20,\dots,50$.
Every column contains the results for different vector $\vec k$.
The first row contains the results for a symmetric Beta distribution, that is $\mathcal{B}(5,5)$.
The second row contains the results for the triangular distribution $\mathcal{T}$. 
The last row contains the results for the Uniform distribution $\mathcal{U}$.\label{fig:battleofvectors_asym_avg}}
\end{figure}



\begin{figure}
  \centering
  
  \begin{subfigure}{0.31\linewidth}
    \includegraphics[width=\linewidth]{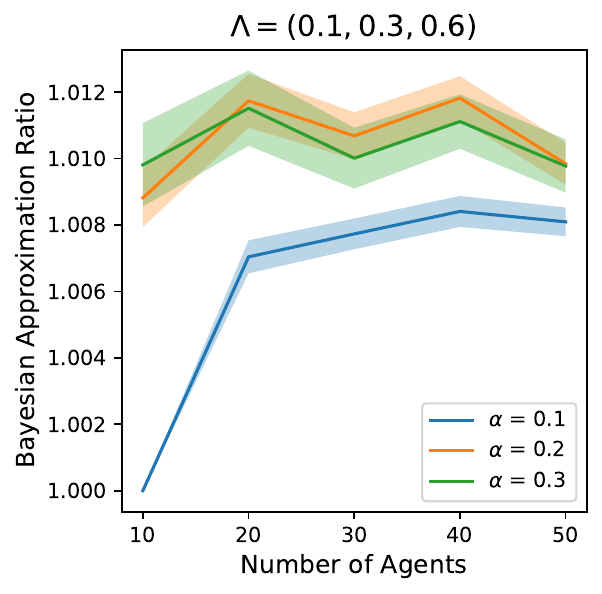}
  \end{subfigure}
  \begin{subfigure}{0.31\linewidth}
    \includegraphics[width=\linewidth]{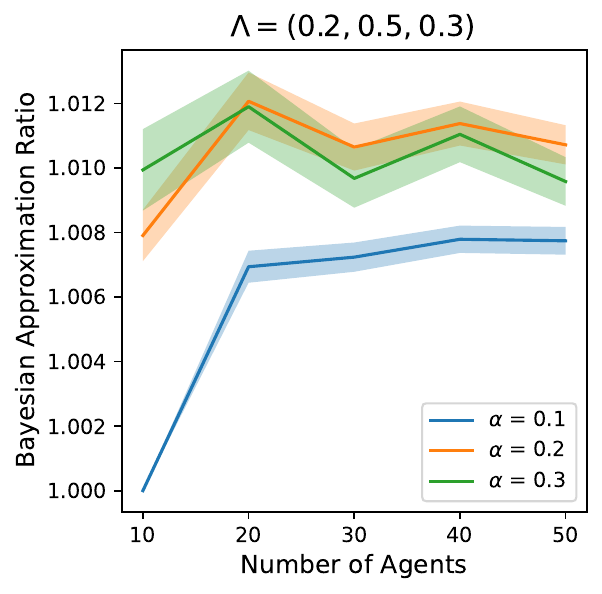}
  \end{subfigure}
  \begin{subfigure}{0.31\linewidth}
    \includegraphics[width=\linewidth]{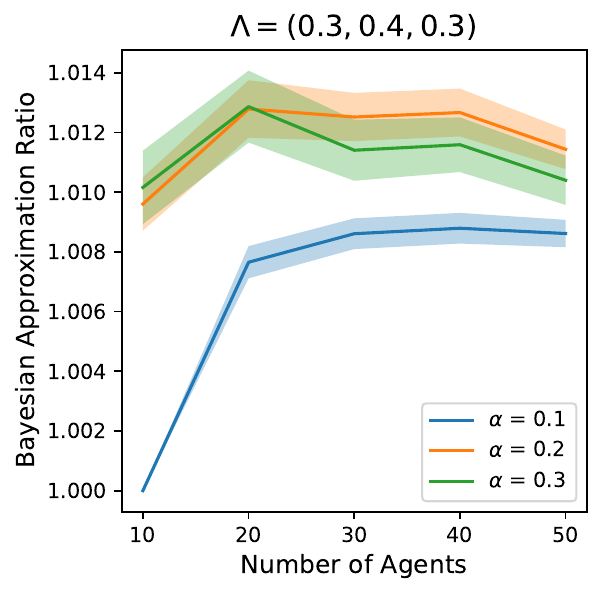}
  \end{subfigure}
  \medskip

    \begin{subfigure}{0.31\linewidth}
    \includegraphics[width=\linewidth]{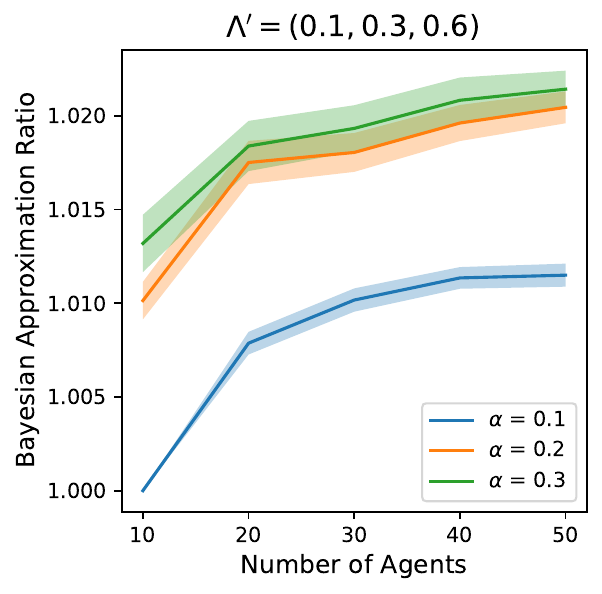}
  \end{subfigure}
  \begin{subfigure}{0.31\linewidth}
    \includegraphics[width=\linewidth]{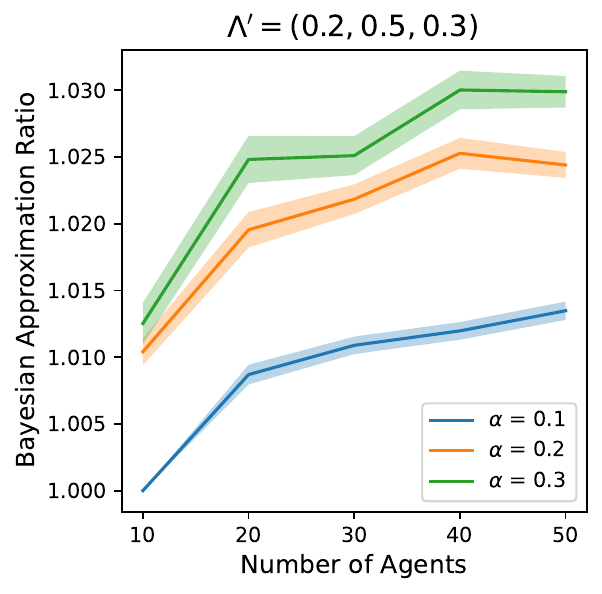}
  \end{subfigure}
  \begin{subfigure}{0.31\linewidth}
    \includegraphics[width=\linewidth]{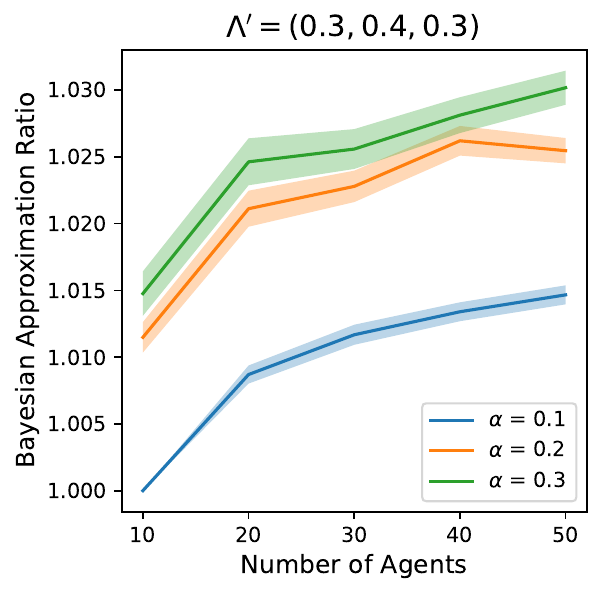}
  \end{subfigure}
\caption{The Bayesian approximation ratio of $\mathcal{PM}_{best}$ for a population non i.d.. The capacities of the facilities are balanced, i.e. $k_1=k_2=\alpha n$ with $\alpha=0.1,0.2,0.3$, and for $n=10,20,\dots,50$.
In the first raw, the Beta distribution is symmetric, in particular $\mathcal{B}(5,5)$, in the second raw the Beta distribution is asymmetric, in particular $\mathcal{B}(1,9)$.
Every column contains the results for different $\Lambda$.\label{fig:battleofvectors_firsttest}}
\end{figure}

  \begin{figure}
  \centering
  
  \begin{subfigure}{0.31\linewidth}
    \includegraphics[width=\linewidth]{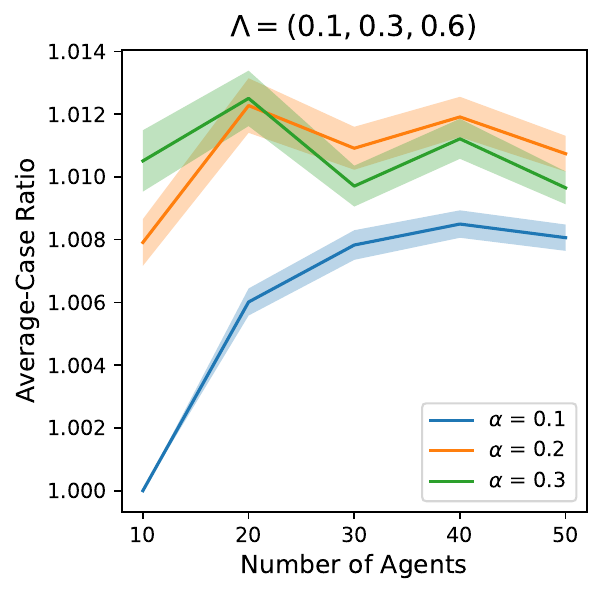}
  \end{subfigure}
  \begin{subfigure}{0.31\linewidth}
    \includegraphics[width=\linewidth]{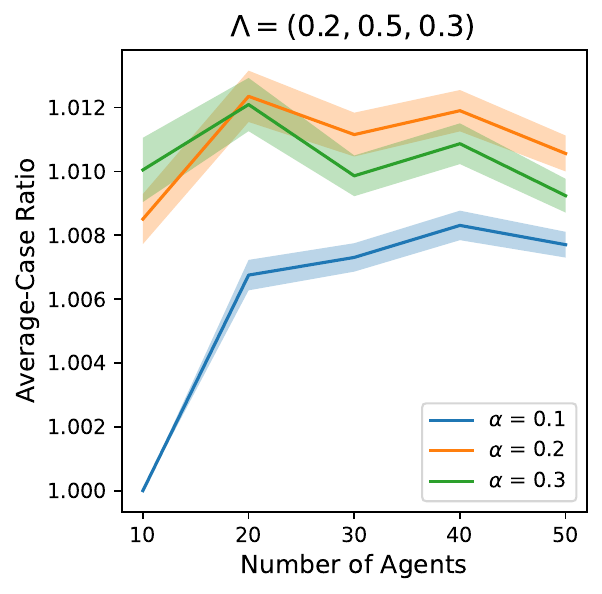}
  \end{subfigure}
  \begin{subfigure}{0.31\linewidth}
    \includegraphics[width=\linewidth]{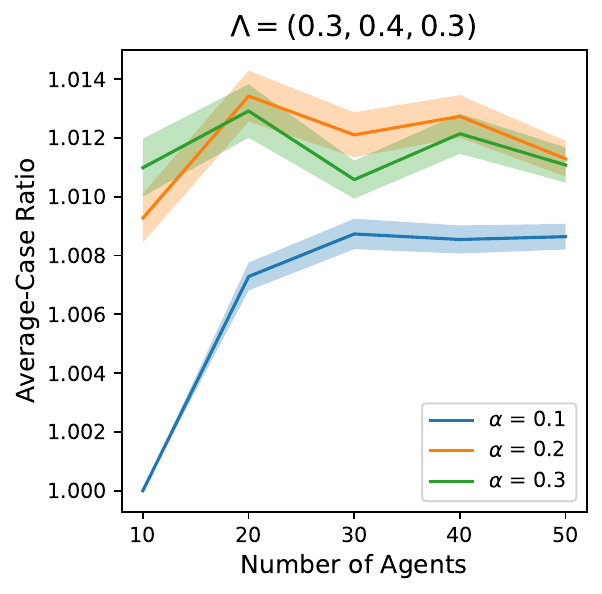}
  \end{subfigure}
  \medskip

    \begin{subfigure}{0.31\linewidth}
    \includegraphics[width=\linewidth]{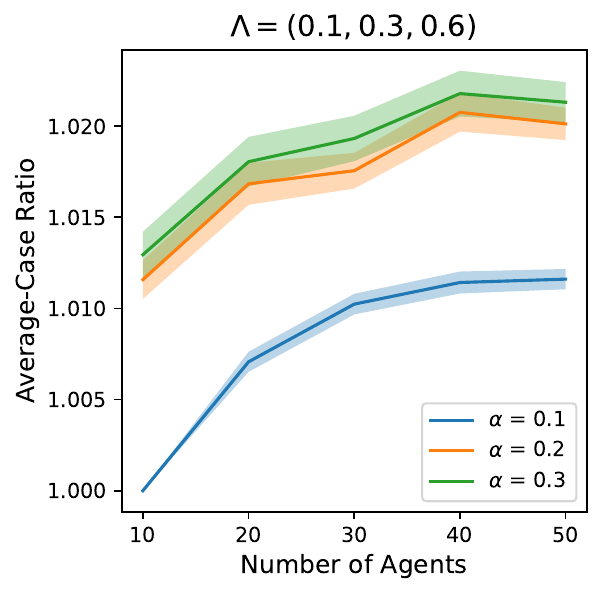}
  \end{subfigure}
  \begin{subfigure}{0.31\linewidth}
    \includegraphics[width=\linewidth]{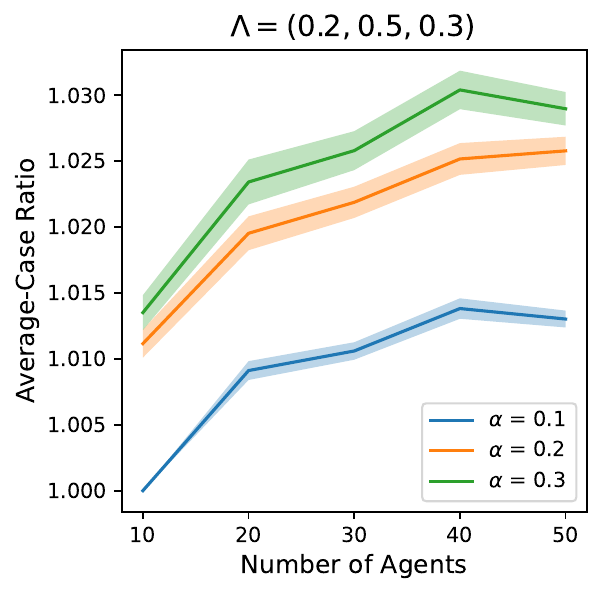}
  \end{subfigure}
  \begin{subfigure}{0.31\linewidth}
    \includegraphics[width=\linewidth]{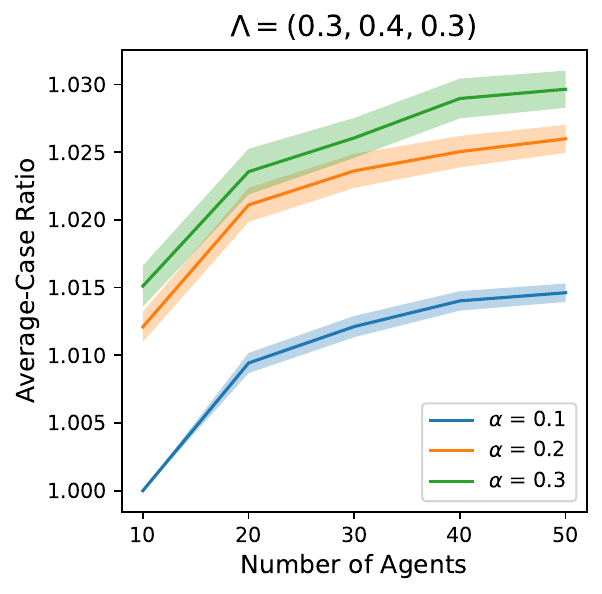}
  \end{subfigure}
\caption{The Average-Case approximation ratio of $\mathcal{PM}_{best}$ for a population non i.d.. The capacities of the facilities are balanced, i.e. $k_1=k_2=\alpha n$ with $\alpha=0.1,0.2,0.3$, and for $n=10,20,\dots,50$.
In the first raw, the Beta distribution is symmetric, in particular $\mathcal{B}(5,5)$, in the second raw the Beta distribution is asymmetric, in particular $\mathcal{B}(1,9)$.
Every column contains the results for different $\Lambda$.\label{fig:battleofvectors_firsttest_avg}}
  \end{figure}



\subsection{Experiment results -- Bayesian approximation ratio for non i.d. populations.} 
In this test, we empirically evaluate the Bayesian approximation ratio of the percentile mechanisms identified by Theorem \ref{thm:bestPMPmechanism} when agents are not identically distributed.
In particular, we consider the case in which each agent of the population is distributed according to $\mathcal{U}$, $\mathcal{T}$, and $\mathcal{B}(5,5)$.
Every instance is identified by the percentages of agents following each distributions, hence we set $\Lambda=(\lambda_U,\lambda_B,\lambda_T)$, where $\lambda_U=\frac{n_U}{n}$, $\lambda_B=\frac{n_B}{n}$, and $\lambda_T=\frac{n_T}{n}$, $n$ is the total number of agents, and $n_U$, $n_B$, and $n_T$ are the number of agents following the uniform, Beta, and Triangular distribution, respectively.
In Figure \ref{fig:battleofvectors_firsttest}, we report our results for different vectors $\Lambda$.
We consider the case in which the capacities of the facilities are balanced, that $k_1=k_2=\floor{\alpha n}$, where $\alpha=0.1$, $0.2$, and $0.3$.
From our experiments we observe that the percentile mechanism achieves an almost optimal Bayesian approximation ratio (peaking at $1.01$), that it is constant regardless of $n$, and that the CI is small (around $0.003$).
Our experiments confirm that the best percentile mechanism according to the worst-case analysis behave almost optimally in a Bayesian framework.
In Figure \ref{fig:battleofvectors_firsttest_avg}, we report our results for the Average-Case approximation ratio.
We observe that, while the expected approximation ratio of the mechanism is still low and close to $1$, the confidence interval in this case is tighter, which further confirms the suitability of the percentile mechanism to handle instances in which the population is non identically distributed.

\section{Conclusion and Future Works}

In this paper, we studied the mechanism design aspects of the $m$-CFLP under the assumption that the total capacity of the facility is smaller than the number of agents to accommodate.
We assume that, after the position of the facility is fixed, the agents compete in a First-Come-First-Served (FCFS) game to gain access to the facilities.
Our main contribution consist in studying the case in which $m\ge 2$, which was left as an open questions in the paper introducing the problem \cite{aziz2020capacity}.
Our approach emphasizes the significance of absolutely truthful mechanisms, which prevent agents from benefiting regardless of their strategy in the FCFS game, and ES mechanisms, whose SW remains independent of the FCFS game equilibrium.
We show that the percentile mechanisms \cite{sui2013analysis} are absolutely truthful and characterize under which conditions they are ES.
We show that ES percentile mechanisms achieve bounded approximation ratio for every $m>1$ and characterize the best percentile vector as a function of $n$, $k_1$, and $k_2$.
Interestingly, if $n>(2k-1)m$, the approximation ratio of the best percentile mechanism $1+\frac{1}{2m-1}$, i.e. is asymptotically optimal with respect to the number of facilities.
Lastly, we run extensive numerical results to study the performances of the percentile mechanism from a Bayesian perspective.
In our future works, we aim at extending this problem to the case in which the agents are located on a generic graph.
Another interesting research avenue is to study how changing the preferences of the agents affects the performances of the mechanisms.
Finally, it would be interesting to study the asymptotic Bayesian approximation ratio of ES percentile mechanisms.
In that regard, it is worth to notice that the upper bound we presented in \eqref{eq:futurework} can be connected to the truncated Wasserstein Distance presented in \cite{Pele2009,auricchio2019maximum}, which enables to apply the same techniques developed in \cite{auricchio2023extended}.

\subsubsection*{Acknowledgements}
Jie Zhang was partially supported by a Leverhulme Trust Research Project Grant (2021 -- 2024) and the EPSRC grant (EP/W014912/1).
\bibliography{sn-bibliography}

\end{document}